\documentclass[reqno,12pt,letterpaper]{amsart}
\usepackage{amsmath,amssymb,amsthm,graphicx,mathrsfs,url}
\usepackage[usenames,dvipsnames]{xcolor}
\usepackage[colorlinks=true,linkcolor=Red,citecolor=Green]{hyperref}
\usepackage{amsxtra}
\usepackage{tikz}
\usepackage{amscd}
\usetikzlibrary{arrows.meta}
\usepackage[normalem]{ulem}
\usepackage{wasysym} 
\usepackage{graphicx}
\usepackage{subcaption}
\usepackage[normalem]{ulem}
\usepackage{a4wide,color,eucal,enumerate,mathrsfs}
\usepackage[normalem]{ulem}
\usepackage{amsmath,amssymb,epsfig,amsthm} 
\usepackage[latin1]{inputenc}
\usepackage{psfrag}
\def\arXiv#1{\href{http://arxiv.org/abs/#1}{arXiv:#1}}
\usepackage{array}
\newcolumntype{P}[1]{>{\centering\arraybackslash}m{#1}}
\usepackage{comment}

\def\?[#1]{\textbf{[#1]}\marginpar{\Large{\textbf{??}}}}

\def\smallsection#1{\smallskip\noindent\textbf{#1}.}
\let\epsilon=\varepsilon 
\newcommand{\R}{{\mathbb R}}
\newcommand{\RR}{{\mathbb R}}

\newcommand{\CC}{{\mathbb C}}

\newcommand{\ZZ}{{\mathbb Z}}

\newtheorem{theo}{Theorem}
\newtheorem{prop}{Proposition}[section]	
\newtheorem{defi}[prop]{Definition}

\newtheorem{assumption}{Assumption}
\newtheorem{lemm}[prop]{Lemma}
\newtheorem{corr}[prop]{Corollary}
\newtheorem{rem}{Remark}
\numberwithin{equation}{section}

\DeclareMathOperator{\Spec}{Spec}

\let\Re=\Real

\DeclareMathOperator{\supp}{supp}

\DeclareMathOperator{\tr}{tr}

\def\indic{\operatorname{1\hskip-2.75pt\relax l}}
\usepackage{scalerel}

\newcommand\reallywidehat[1]{\arraycolsep=0pt\relax%
\begin{array}{c}
\stretchto{
  \scaleto{
    \scalerel*[\widthof{\ensuremath{#1}}]{\kern-.5pt\bigwedge\kern-.5pt}
    {\rule[-\textheight/2]{1ex}{\textheight}} 
  }{\textheight} %
}{0.5ex}\\           
#1\\                 
\rule{-1ex}{0ex}
\end{array}
}

\def\tr{\operatorname{tr}}

\def\bint{{\ifinner\rlap{\bf\kern.35em--}
\int\else\rlap{\bf\kern.45em--}\int\fi}\ignorespaces}

\def\bbint{{\ifinner\rlap{\bf\kern.35em--}
\hspace{0.078cm}\int\else\rlap{\bf\kern.45em--}\int\fi}\ignorespaces}

\title[Semiclassical TMDs]{Twisted TMDS in the small-angle limit: \\ Exponentially flat and trivial bands} 

\author{Simon Becker}
\email{simon.becker@math.ethz.ch}
\address{ETH Zurich, 
Institute for Mathematical Research, 
R\"amistrasse 101, 8092 Zurich, 
Switzerland}

\author{Mengxuan Yang}
\email{mxyang@math.berkeley.edu}
\address{Department of Mathematics, University of California, Berkeley; 970 Evans Hall, Berkeley, CA, 94720-3840}

\begin{document}
\maketitle
\begin{abstract}
Recent experiments discovered fractional Chern insulator states at zero magnetic field in twisted bilayer MoTe$_2$ \cite{C23,Z23} and $WSe_2$\cite{MD23}. In this article, we study the MacDonald Hamiltonian for twisted transition metal dichalcogenides (TMDs) and analyze the low-lying spectrum in TMDs in the limit of small twisting angles. 
Unlike in twisted bilayer graphene Hamiltonians, we show that TMDs do not exhibit flat bands. 
The flatness in TMDs for small twisting angles is due to spatial confinement by a matrix-valued potential. 
We show that by generalizing semiclassical techniques developed by Simon \cite{Si83} and Helffer-Sj\"ostrand \cite{HS84} to matrix-valued potentials, there exists a wide range of model parameters such that the low-lying bands are of exponentially small width in the twisting angle, topologically trivial, and obey a harmonic oscillator-type spacing with explicit parameters. 
\end{abstract}

\section{Introduction}
The field of twistronics is the study of how the twisting angle between layers of two-dimensional materials changes their electrical properties. The first material that substantially boosted the field was twisted bilayer graphene where at certain twisting angles, the magic angles, the material exhibits superconductivity for certain fillings. Soon the field of moir\'e materials expanded to many different classes of materials, cf. \cite{TLF22,TPF22}.
More recently, a class of two-dimensional periodic semiconductors so-called transition metal dichalcogenides (TMDs) when stacked and twisted appropriately exhibited the fractional anomalous quantum Hall effect \cite{N11,R11,S11} at small angles. The purpose of this paper is to initiate a first mathematical study of the effective one-particle Hamiltonian of twisted TMDs. 

As we will see, the Hamiltonians for twisted bilayer graphene (TBG) and twisted TMDs are mathematically of quite different flavour. The former is a semiclassical Dirac operator while the latter is a semiclassical Schr\"odinger operator and therefore much more amenable to semiclassical techniques, since the semiclassical Dirac operator is unbounded from above and below while \emph{the magic} happens at zero energy. 

 Their physical properties are described by matrix-valued Schr\"odinger-type operators with complex-valued potentials. Due to the intricate structure of such operators, fine spectral theoretic properties are analytically intractable. Thus, there is a demand to develop techniques to obtain an effective spectral description in a semiclassical limit, with the twisting angle between the two lattice structures as the semiclassical parameter explaining e.g. the localized states at low band energies observed in \cite{J19}.

The Hamiltonian for twisted TMDs is a semiclassical Schr\"odinger operator 
\[ H = -h^2\Delta \operatorname{Id}_{\mathbb C^2}+ V,\]
where $V \in C^{\infty}(\mathbb R^2/\Gamma;\mathbb C^{2\times 2})$ is a specific periodic Hermitian matrix-valued potential.
The semiclassical parameter $h$ is the mechanical twisting angle. The precise form of the potential matrix is given in \eqref{eq:semiclassique}. Unlike the TBG Hamiltonian, TMD Hamiltonians and more broadly Schr\"odinger operators with matrix-valued potential do in general not exhibit flat bands. We state the result for lattices $\Gamma$ in $\mathbb R^n$ but restrict us to $\mathbb R^2$ in our proof for simplicity. 
\begin{theo}
\label{theo:no_flat}
Let $\Gamma$ be a lattice in $\mathbb R^n$. The Hamiltonian $H_V=-\Delta + V$, where $V \in L^{\infty}(\mathbb R^n/\Gamma;\mathbb C^{n \times n})$ is Hermitian, does not have any flat bands, i.e., there does not exist a $\lambda$ such that $\lambda \in \Spec_{L^2(\mathbb R^2/\Gamma)}(H_{V,\mathbf k})$ for all $\mathbf k$ with 
\[ H_{V,\bf k}:=(D_x +{\bf k})^2 \operatorname{Id}_{\CC^n} + V: H^2(\mathbb R^n/\Gamma) \to L^2(\mathbb R^n/\Gamma),\quad D_{x}:=-i \nabla_{x}.\]
\end{theo}

In the case of twisted TMDs one finds that the two eigenvalues $\lambda_\pm(V)$ of the potential matrix $V\in C^{\infty}(\mathbb R^2/\Gamma;\mathbb C^{2\times 2})$ are gapped, i.e., 
\[ \lambda_-(V(x)) <\lambda_+(V(x)) \text{ for all } x \in \mathbb R^2/\Gamma,\] 
where $\Gamma = v_1 \ZZ\oplus v_2 \ZZ$ is the rescaled moir\'e lattice with 
\[v_1 =-\frac{1}{2} \begin{pmatrix}\sqrt{3},1 \end{pmatrix}^t, \quad v_2 =\frac{1}{2} \begin{pmatrix}\sqrt{3},-1 \end{pmatrix}^t, \text{ and }v_3= (0,1)^t. \]

In this work we make the following assumption which holds for a large range of model parameters.
\begin{assumption}
\label{ass1}
    We assume that $0$ is the unique non-degenerate global minimum of $\lambda_-(V)$, modulo translation by lattice $\Gamma$. 
\end{assumption}
While this assumption is not strictly necessary for our methods, the energy scale at which our results hold is only meaningful when this assumption is met, see e.g., Figure \ref{fig:2} for an illustration of possible energy landscapes of $\lambda_-(V(x)).$ 
Future work will be devoted to the full parameter range which requires a more intricate analysis. 

Since $H$ is a periodic operator, we can apply the Bloch transform and study Bloch operators 
\[ H_{\mathbf k} = h^2 (D_x + \mathbf k)^2 \operatorname{Id}_{\CC^2} + V : H^2(\RR^2/\Gamma;\CC^{2\times 2}) \to L^2(\RR^2/\Gamma;\CC^{2\times 2}), \]
where $\operatorname{Id}_{\CC^2}$ is the identity on $\CC^2$ and $\mathbb C^{2\times 2}$ is dropped when there is no confusion hereafter. One then wants to compare the eigenvalues $E_n(h,\mathbf k)$ of the Bloch operators to the eigenvalues of a suitable operator $ {H}_{\text{well}}$ that exhibits only a single potential well at $0$ but not at any of its lattice translates. One thus defines
\[ {H}_{\text{well}}=H + (1-\chi)\operatorname{Id}_{\CC^2}:\ H^2(\RR^2) \to L^2(\RR^2) \]
$C_c^{\infty}(\RR^2;[0,1])\ni \chi $ is equal to $1$ in a neighbourhood to $0$ and vanishes in a neighbourhood of all $\Gamma \setminus \{0\}.$
We then have the following result which includes twisted TMD model parameters $\phi,\beta, U$ (cf.~equation \eqref{eq:semiclassique}). 
\begin{theo}
\label{theo:harmonic_approximation}
Let $(E_n(h))$ be an increasing enumeration of eigenvalues of $H_{\text{well}}$ below the essential spectrum. Let $(E'_n)_{n \ge 1}$ be an increasing enumeration of the spectrum of the operator $H_{\text{pert}}$ given by (cf.~Lemma \ref{l:harmonic})
\[ E_n' = \lambda_n h + \lambda_{-}(0) + \mathcal O_n(h^{3/2}), \text{ with } \lambda_{-}(0)=6 \cos(\phi) - \frac{ \sqrt{9 \beta^2 + U^2} }{2},\] 
where $\lambda_n$ the n-th smallest element counting multiplicity in  
$$\mathcal A=\left\{\sum_{i=1}^2 (2m_i+1) \sqrt{\lambda_{i,-}}, m\in \mathbb N_0^2 \right\}$$ 
and constants  
\begin{equation*}
    \begin{split}
\lambda_{1,-}&=-8\pi^2  \cos(\phi) +
  \frac{2\beta ^2 \pi^2}{3\sqrt{9\beta^2 + U^2}}, \text{ and }\\
 \lambda_{2,-}&=-8\pi^2 \cos(\phi) +
   \frac{6\beta ^2  \pi^2  }{\sqrt{9\beta^2 + U^2}} .
\end{split}
\end{equation*}
Then it follows that
\begin{equation}
    \lim_{h \downarrow 0} \frac{E_n(h)- E'_n}{h}=0.
\end{equation}
Similarly, let $H_{\mathbf k}$ be the Bloch operator with Bloch eigenvalues $ (E_n(h,\mathbf k))_n$, then
\begin{equation}
    \lim_{h \downarrow 0} \sup_{\mathbf k \in \mathbb R^2/\Gamma^*} \frac{\vert E_n(h,\mathbf k)-E_n'\vert}{h}=0.
\end{equation}
\end{theo}

We emphasize the difference to the case of TBG, where in the small angle limit, bands exhibit an exponentially small spacing (cf.~\cite{BEWZ22}) close to zero energy.

\begin{theo}
\label{theo:triv_bands}
    Let $H_{\mathbf k}$ be as above, then the Chern number associated with the lowest $N$ Bloch bands vanishes for $h>0$ small enough.
\end{theo}

Finally, we show that the bands are exponentially small in size and close to the eigenvalues $E_n(h)$ of $H_{\text{well}}.$

\begin{theo}
\label{theo:exp_narrow_bands}
Let $N$ be fixed. There exists a constant $C>0$ such that uniformly in $\mathbf k \in \mathbb R^2/\Gamma^*$ for all $1\le i \le N$
\[ E_i(h,\mathbf k)= E_i(h) + \mathcal O(e^{-C/h}). \]
\end{theo}

\subsection{Agmon estimates, tight-binding limits \& related work}
To put our work into context, we want to mention some related works in the field of scalar Schr\"odinger operators. 
Agmon estimates \cite{A79}, see also Carmona--Simon \cite{CS81}, are by now classical estimates to control the decay of eigenfunctions in the classically forbidden region using a pseudo-metric, the Agmon distance $d_E:=(V(x)-E)_+ \ dx.$ This quantity measures the distance to the classical region defined by an energy threshold $E$. The importance of Agmon estimates in semiclassical analysis has been pointed out by Simon \cite{S83} and Helffer-Sj\"ostrand \cite{HS84} for the double well problem. In the simplest case, one has
\begin{theo}\cite[Theo.\,$3.4$]{HS12}
Let $H=-\Delta + V$ with $V$ continuous and real such that $H$ is closed with $\Spec(H)\subset \mathbb R$.  Then for $E \in \Spec(H)$ an eigenvalue such that $\supp(\{x; (E-V(x))_+\})$ is compact, we have 
\[\int_{\mathbb R^n} e^{2(1-\varepsilon)d_E(x,0)} \vert \psi(x)\vert^2 \ dx \le c_{\varepsilon} \text{ for all }\varepsilon>0. \]
\end{theo}
Thus, for $V \in C^{\infty}(\mathbb R^n)$ with $r_0:=\liminf_{\vert x \vert \to \infty} V > \inf V=0$, we have discrete spectrum below the essential spectrum $[r_0,\infty)$ of the semiclassical Schr\"odinger operator $H=-h^2 \Delta + V(x)$ by Persson's theorem (see also \cite[Theorem 7.3]{zw12}). The decay of the associated eigenfunctions is controlled as described by the theorem. The Agmon estimates are in general not optimal and lower bounds may differ substantially, as pointed out by Erd\"os \cite{E96} who obtained Gaussian decay estimates for magnetic eigenfunctions below the essential spectrum. 

In case that $V$ exhibits infinitely many non-degenerate wells situated at points $\Gamma$, Carlsson \cite{C90} proved a general result about the characterization of the low-lying spectrum of semiclassical Schr\"odinger operators with infinitely many wells in terms of an infinite matrix acting on $\ell^2(\Gamma)$. Carlsson's result rests on essentially three ingredients: 
\begin{enumerate}
\item The exponential decay of the single well wavefunctions by filling all wells but one, 
\item controlling the density of wells by making appropriate assumptions limiting the accumulation of wells,  
\item a control of norms of various infinite matrices.
\end{enumerate}

For the special case when $\Gamma$ is a lattice and $V$ a periodic potential with respect to $\Gamma$, i.e. $V(x + \gamma)=V(x)$ for all $x \in \mathbb R^n$ and $\gamma \in \Gamma$, Carlsson's infinite matrix corresponds to a convolution operator which shows that the first band of the semiclassical problem is exponentially close to the ground-state energy of the \emph{one-well} reference Hamiltonian. This has been extended by Klopp \cite{K91} to perturbations of the periodic scenario. 

Largely independent of the previous developments, but still following the idea of \emph{filling the wells}, Fefferman, Lee-Thorp, and Weinstein \cite{FLW17} studied the tight-binding limit for the honeycomb lattice and obtained stronger and more explicit convergence results. In particular, they obtained a resolvent norm convergence \cite[Theo. $6.2$]{FLW17} for a rescaled version of the semiclassical Floquet-Bloch Hamiltonian $\tilde H_{\mathbf k}$ and the corresponding tight-binding Hamiltonian $H^{\text{TB}}_{\mathbf k}$
\[(\tilde H_{\mathbf k}-\lambda)^{-1} - J (H^{\text{TB}}_{\mathbf k}-\lambda)^{-1} J^*  = \mathcal O(e^{-c/h}), \]
with $J$ maps lattice points to appropriate $L^2$ Bloch functions. 
This approach, using the Schur complement formula guarantees that topological properties of the band structure are preserved in the tight-binding limit. Therefore, the approach by Fefferman, Lee-Thorp, and Weinstein also justifies the use of tight-binding models for the study of topological properties.

The graphene Hamiltonian has also been analyzed, along with the triangular lattice, in the PhD thesis of Kerdelhu\'e \cite{Ke91} for weak magnetic fields, building up on the work of Carlsson, where he derives the tight-binding limit and obtains first results on the self-similar structure on the spectrum.

In the presence of magnetic fields, a smallness condition on the magnetic field has been usually imposed to overcome the non-real valued ground states for general magnetic fields. This has been partly overcome in the works by Fefferman, Shapiro, and Weinstein \cite{FSW22} and Shapiro-Weinstein \cite{SW22} where an additional radial assumption allows them to obtain a lower bound on the tunnelling coefficient in the presence of comparably large magnetic fields, see also Helffer-Kachmar \cite{HK23}. Fefferman, Shapiro, and Weinstein also announce forthcoming results showing that the radial assumption is a necessary assumption to obtain a lower bound on the tunnelling coefficient. 

We believe that this strategy here can also be used to study twisted TMDs satisfying Assumption \ref{ass1} in at least constant magnetic fields. 

A full tight-binding reduction for matrix-valued potentials seems desirable to get a precise understanding of the band structure in the semiclassical limit than we do. While our approach gives a good understanding when Assumption \ref{ass1} holds, going beyond this restriction requires a more intricate analysis and will be addressed in future work. 

Outline of the paper 
\begin{itemize}
\item In Section \ref{sec:Hamiltonian}, we introduce the continuum model for twisted TMDs and show that it does not exhibit any flat bands, i.e., proof Theorem \ref{theo:no_flat}.
\item In Section \ref{sec:harmony}, we roughly locate the low-lying spectra of $H_{\text{well}}, H_{\mathbf k}$ using a harmonic approximation and prove Theorem \ref{theo:harmonic_approximation}.
\item In Section \ref{sec:topology}, we prove Theorem \ref{theo:triv_bands}.
\item In Section \ref{sec:Agmon}, we introduce Agmon distances and decay estimates tailored to our operator.
\item In Section \ref{sec:exp_width} we prove Theorem \ref{theo:exp_narrow_bands} and show that bands are exponentially narrow in the twisting angle as $h \downarrow 0$.
\end{itemize}

\smallsection{Notations}
We denote by $D_{x}:=-i \nabla_{x}$ the self-adjoint momentum operator and $D_{x_i}=-i\partial_{x_i}$, accordingly. By $\sigma_i$, we  denote the $i$-th Pauli matrix with 
\[ \sigma_1 = \begin{pmatrix} 0 & 1 \\ 1 & 0 \end{pmatrix}, \quad \sigma_2 = \begin{pmatrix} 0 & - i \\ i & 0 \end{pmatrix}\text{ and }
\sigma_3 = \operatorname{diag}(1,-1).
\]
For a self-adjoint matrix $A \in \CC^{n \times n}$ we denote by $\lambda_1(A) \le...\le \lambda_n(A)$ the eigenvalues. The standard basis vectors are denoted by $\hat{e}_1,...,\hat{e}_n.$
The identity matrix on $\mathbb C^{n}$ is denoted by $\operatorname{Id}_{\mathbb C^n}.$

\section{The continuum Hamiltonian for twisted transition metal dichalcogenides}
\label{sec:Hamiltonian}
Transition metal dichalcogenides are semiconductors. The single valley Hamiltonian for holes in the valence band at twisting angle $\theta$ is \cite{W19,DCZ21,D23}
\[ H_{\text{valley}}= \begin{pmatrix} -\Delta_x + \alpha V_{\uparrow}(\theta x)   & \beta T(\theta x) \\ \beta T(\theta x)^* & -\Delta_x  + \alpha V_{\downarrow}(\theta x) \end{pmatrix} + \sigma_3 U/2,\]
where $U>0$ is the displacement field, the moir\'e intralayer potential is given by
\[\begin{split}V_{\uparrow/\downarrow}(x) 
&= 2\sum_{j=1}^3 \cos\Big(\frac{4\pi}{\sqrt{3}}\langle R^{2(j-1)} e_1 ,  x \rangle \pm \phi\Big) \\
&= 2 \Big(\cos( \tfrac{4 \pi x_1}{\sqrt{3}} \pm \phi) + 
   2 \cos( \tfrac{2 \pi x_1}{\sqrt{3}}\mp \phi) \cos(2 \pi x_2)\Big) \end{split} \]
and the interlayer tunnelling potential by 
\[ \begin{split} T(x) &=   1 + e^{\frac{4\pi i}{\sqrt{3}} \langle Re_1 , x \rangle } + e^{\frac{4\pi i}{\sqrt{3}}\langle R^2e_1,x \rangle} \\
&= 1 + e^{-\frac{2 i \pi x_1}{\sqrt{3}} - 2 i \pi x_2} + e^{-\frac{2 i \pi x_1}{\sqrt{3}} + 2 i \pi x_2} \\
  &=1+2e^{-\frac{2 i \pi x_1}{\sqrt{3}}} \cos(2\pi x_2)
    \end{split}\]
with $R = \tfrac{1}{2} \begin{pmatrix} - 1 & -\sqrt{3} \\ \sqrt{3}& -1 \end{pmatrix}$ the rotation by $2\pi/3,$ and $\alpha,\beta>0$ are coupling parameters. 

To reduce this Hamiltonian to unit length scales, we perform the substitution $x\theta \mapsto x$ and obtain the semiclassical Hamiltonian (in the small twisting angle regime $\theta\ll 1$) 
\begin{equation}
\label{eq:semiclassique}
 H := \begin{pmatrix} -h^2 \Delta_x + \alpha V_{\uparrow}( x)   & \beta T( x) \\ \beta T( x)^* & -h^2 \Delta_x  + \alpha V_{\downarrow}( x) \end{pmatrix} + \sigma_3 U/2 = -h^2\Delta_x\cdot \mathrm{Id}_{\CC^2} + V
 \end{equation}
where we define $h:=\theta$ to be the semiclassical parameter.
This Hamiltonian is periodic with respect to the moir\'e lattice $\Gamma = v_1 \ZZ\oplus v_2 \ZZ$ with 
\[v_1 =-\frac{1}{2} \begin{pmatrix}\sqrt{3},1 \end{pmatrix}^t, \quad v_2 =\frac{1}{2} \begin{pmatrix}\sqrt{3},-1 \end{pmatrix}^t, \text{ and }v_3= (0,1)^t, \]
where $R v_{i}=v_{i+1}$ where $R$ is the rotation by $2\pi/3.$

\subsection{Symmetries \& Absence of flat bands}
We notice the symmetries $V_{\uparrow/\downarrow}(x+q)=V_{\uparrow/\downarrow}(x)$ and $T(x+q)=T(x)$ for $q \in \Gamma=v_1 \ZZ +v_2 \ZZ .$
The potential $V_{\uparrow/\downarrow}$ and $T$ has symmetries
\begin{gather*}
    V_{\uparrow/\downarrow}(R x) = V_{\uparrow/\downarrow}(x), \quad V_{\uparrow/\downarrow}(-x) = V_{\downarrow/\uparrow}(x),\\
    T(Rx)=T(x), \quad T(-x) = T(x)^*.
\end{gather*}
Let $\mathcal Ru(x):=u(Rx)$ and $\mathcal T_qu(x):=u(x+q)$ with $q \in \Gamma$, then we have 
\[ \mathcal R H \mathcal R = H \text{ and } \mathcal T_q^* H \mathcal T_q = H.\]
We thus have, for $U=0$, the $C_2T$ symmetry $\mathcal Cu(x) := \sigma_1 \overline{u(-x)}$
\[ \mathcal CH\mathcal C=H. \]
\begin{rem}[Honeycomb-lattice potentials]
When $V_{\uparrow} = V_{\downarrow}$, then $V$ is a honeycomb-lattice potential as studied by Fefferman-Weinstein, see \cite{FW12}. If this condition does not hold, then the model does not exhibit Dirac points in general, see Figure \ref{fig:phi}.
\end{rem}
\begin{rem}[Coordinates for numerics]
For numerical simulations, it is also convenient to perform a change of variables $x_1+i x_2 = \frac{i\omega}{2\pi}( y_1 + \omega y_2)$ with $y_1,y_2 \in \RR/(2\pi \ZZ)$ such that $x_1 =\frac{ \sqrt{3}(y_2-y_1)}{4\pi}$ and $x_2=\frac{-y_1-y_2}{4\pi}$
\[ V_{\text{num},\uparrow/\downarrow}(y)=2\Big(\cos\big( y_1  \pm \phi\big)+ \cos\big( y_2  \mp \phi \big)  + \cos\big( ( y_1-y_2)\mp \phi \big)\Big)\] 
and $T_{\text{num}}(y) =\omega(1 + e^{ -i y_1} + e^{iy_2}).$
\end{rem}
\subsection{Bloch-Floquet theory}

The periodicity of the TMDs Hamiltonian $H: H^2(\RR^2)\to L^2(\RR^2)$ in \eqref{eq:semiclassique} with respect to the lattice $\Gamma$ means that we can study the operator using the Bloch-Floquet-transformed Laplacian $-\Delta_{\mathbf k} = (D_x-\mathbf k)^2$ and 
\[ H_{\mathbf k} := \begin{pmatrix} -h^2\Delta_{\mathbf k}  + \alpha V_{\uparrow}(x)  &\beta T(x) \\ \beta \overline{ T(x)} & -h^2\Delta_{\mathbf k} +  V_{\downarrow}(x) \end{pmatrix}: H^2(\RR^2/\Gamma) \to L^2(\RR^2/\Gamma)\]
for $\mathbf k \in \CC/\Gamma^*,$ where $\Gamma^*$ is the dual lattice of $\Gamma.$
We immediately have from Bloch-Floquet theory that 
\[ \Spec_{L^2(\CC)}( H_{\text{sem}}) = \bigcup_{\mathbf k \in \RR^2/\Gamma^*} \Spec_{L^2(\RR^2/\Gamma)}(H_{\mathbf k}).\]
We refer to \cite[Chapter 5]{notes} for a detailed presentation of Bloch-Floquet theory.

\subsection{Absence of flat bands}
We shall now prove Theorem \ref{theo:no_flat} and show in particular that twisted TMDs do not exhibit flat bands. The original idea of showing absence of flat bands has been introduced by Thomas \cite{T73}.
\begin{proof}[Proof of Theo.\,\ref{theo:no_flat}]
We use the identity
\begin{equation}
    \label{eq:res_idd}
    H_{V,\mathbf k} - \lambda = (-\Delta_{\mathbf k}-\lambda)\left(\operatorname{Id}+(-\Delta_{\mathbf k}-\lambda)^{-1} V \right).
\end{equation}
Thus $\lambda \in \Spec_{L^2(\RR^n/\Gamma)}(H_{V,\mathbf k})$ implies that either $1 \in \Spec_{L^2(\RR^n/\Gamma)}\Big( (\Delta_{\mathbf k} +\lambda)^{-1} V \Big)$ or $\lambda \in \Spec_{L^2(\RR^n/\Gamma)}(-\Delta_{\mathbf k}).$

We now fix $k_2$. If $\lambda$ is a real eigenvalue associated with a flat band, i.e., $\lambda \in \Spec_{L^2(\RR^n/\Gamma)}(H_{V,\mathbf k})$ for all $\mathbf{k}\in \RR^2/\Gamma^*$, then as the compact operator $(\Delta_{\mathbf k} +\lambda)^{-1}V$ is analytic in $k_1$, \cite[Ch\,7, Theo.\,$1.9$]{kato}, there is a non-empty open set
\[N:=\left\{k_1 \in \CC; 1 \in \Spec\Big( (\Delta_{\mathbf k} +\lambda)^{-1} V \Big)\right\}\subset \mathbb{C},\]
as $-\Delta_{\mathbf k}$ does not have flat bands at any energy $\lambda$.
We also define
\[ D:=\left\{k_1 \in \mathbb C,(\Delta_{\mathbf k} +\lambda)^{-1} \text{ is bounded}\right\}.\]
Since holomorphic dependence of the operator $(\Delta_{\mathbf k} +\lambda)^{-1} V$ in $k_1$ is valid for all $k_1 \in D$, we conclude that $ D\subset N$ by \cite[Ch.\,7, Theo.\,$1.9$]{kato} and the fact that $\lambda$ is the flat band energy. This inclusion however is impossible:

We note that for $\mathbf k = ( k_1, k_2)$
 $$- \Delta_{\mathbf k} - \lambda = (D_{x_1} +  k_1)^2 + (D_{x_2} +  k_2)^2 -\lambda.$$
 Thus, by complexifying $k_1 = \mu_1 + i \mu_2$ for $\mu \in \RR^2$, we find 
 \[\Vert (- \Delta_{\mathbf k} - \lambda)^{-1}\Vert = \mathcal O(1/\mu_2^2).\]
 The last estimate is easily verified using the Fourier representation of the differential operator. Indeed, upon conjugating by the Fourier series, the operator $(- \Delta_{\mathbf k} - \lambda)^{-1}$ on $L^2(\mathbb R^2/\Gamma)$ corresponds to multiplication by  
 \[\Big(\frac{1}{(n_1+k_1)^2 + (n_2 + k_2)^2 -\lambda}\Big)_{n \in \Gamma^*} \]
 on $\ell^2(\Gamma^*).$
Hence, choosing $\mu_2 \gg 1$ large enough, the inverses exists and \[\Vert ( \Delta_{\mathbf k} +\lambda)^{-1} V \Vert = \mathcal O(1/\mu_2^2).\] Thus, $1$ cannot be an eigenvalue of this operator. Therefore, $\lambda$ cannot be an eigenvalue associated with flat band.
\end{proof}

The semiclassical Weyl symbol of this Hamiltonian is 
\[ \sigma_h(H) = \begin{pmatrix} \vert \xi\vert^2  + \alpha V_{\uparrow}(x) + \frac{U}{2} & \beta T( x) \\ \beta T( x)^* & \vert \xi\vert^2+ \alpha V_{\downarrow}( x) - \frac{U}{2}  \end{pmatrix}\]

\subsection{Diagonalization of the TMD Hamiltonian}
Setting $\alpha=1$, we observe that the semiclassical symbol of $H$ is 
\[ \sigma_h(H)(x,\xi) = \begin{pmatrix} \vert\xi  \vert^2 + V_{\uparrow}(  x) + U/2 & \beta T( x) \\ \beta T( x)^* & \vert\xi  \vert^2  + V_{\downarrow}( x) -U/2\end{pmatrix}.\]
This matrix has eigenvalues
\begin{equation}
    \label{eq:eigenvalue}
    \begin{split} \lambda_{\pm}(x,\xi) = \vert \xi\vert^2 & + \lambda_{\pm}(V(x)) :=  \vert \xi\vert^2 + \frac{V_{\uparrow}(  x)  + V_{\downarrow}( x) }{2} \pm U_{\text{eff}}(x) \text{ with }\\
U_{\text{eff}}(x)&=  \frac{\sqrt{ (V_{\uparrow}( x) -V_{\downarrow}(x)+ U)^2+ \beta^2\vert T(x)\vert^2}}{2}.\end{split}
\end{equation}
In particular, expanding at $x=0$, we have
\begin{equation}
    \label{eq:potential}
    \lambda_{\pm}(V)= 6 \cos(\phi) \pm \frac{ \sqrt{9 \beta^2 + U^2} }{2} -8\pi^2 \vert x \vert^2 \cos(\phi) \mp
 \frac{2 }{3} \pi^2  \frac{\beta ^2 (x_1^2 + 9 x_2^2)}{\sqrt{9\beta^2 + U^2}} + \mathcal O( \vert x \vert^3)
\end{equation}
 which implies the existence of a non-degenerate potential well for $\lambda_-$ at least for $\phi \in (\pi/2,3\pi/2).$
\begin{figure}

\includegraphics[width=7.5cm]{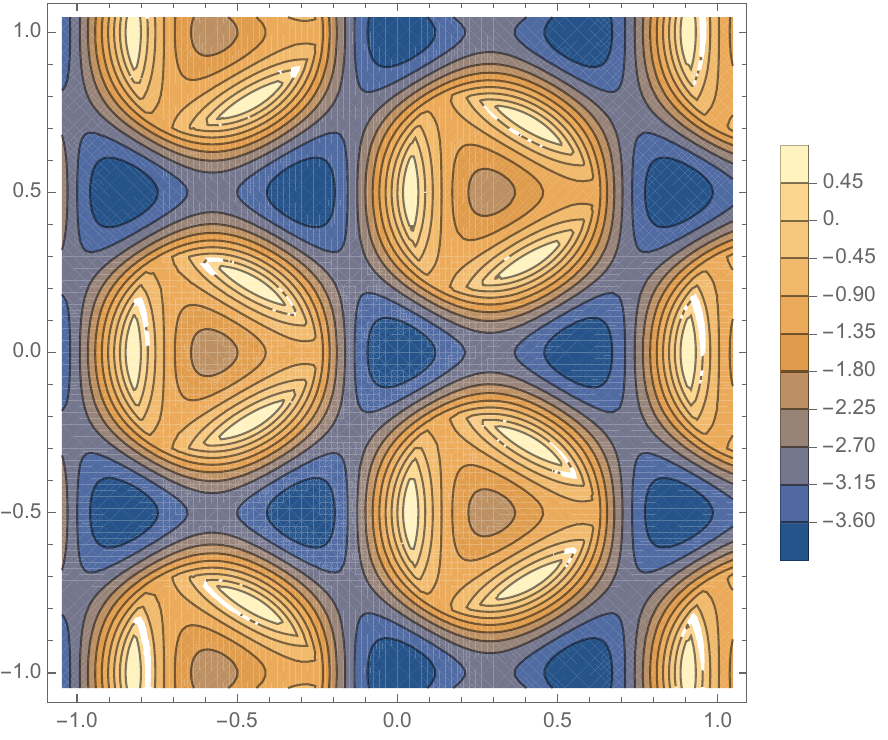}
\includegraphics[width=7.5cm]{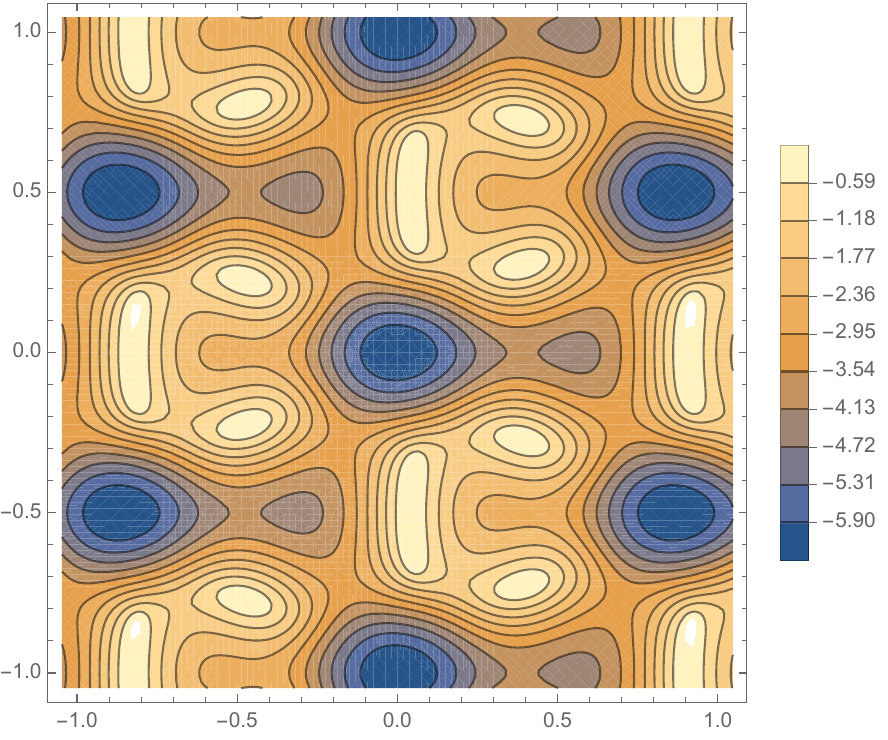}\\
\includegraphics[width=7.5cm]{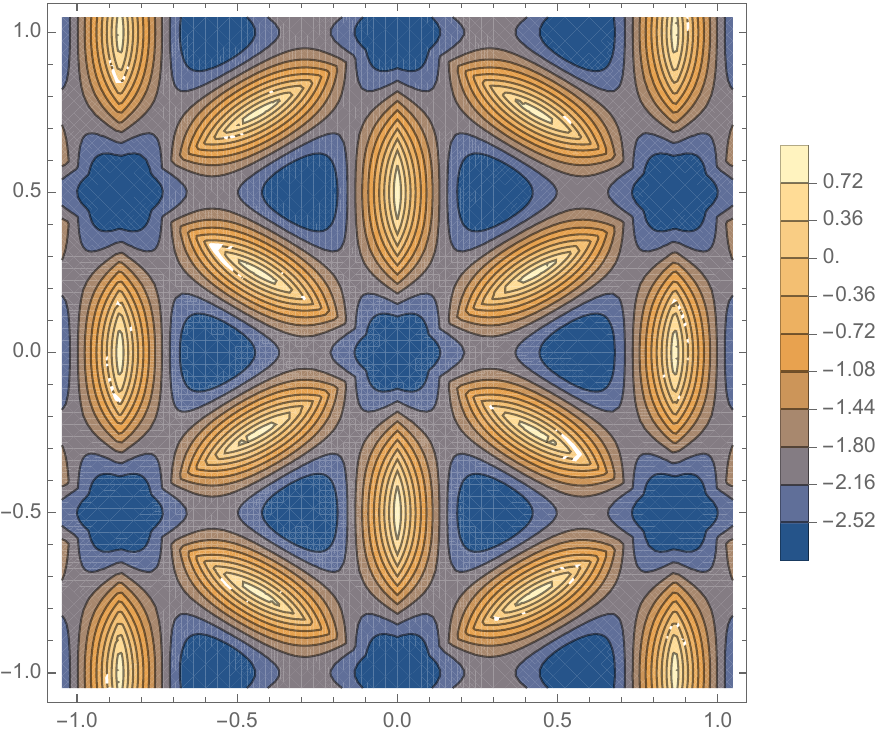}
\includegraphics[width=7.5cm]{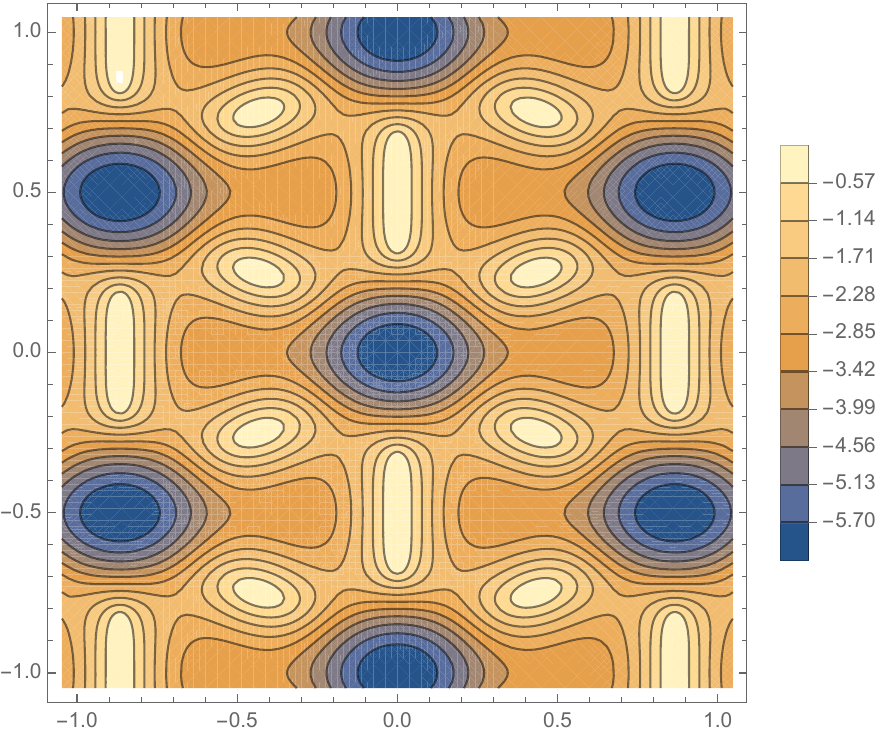}
\caption{ \label{fig:1}Contour plot of $x \mapsto \lambda_-(V(x))$ with $(U,\beta)=(2,0),(2,5),(0,5),(0,0)$ from top left (clock-wise) with $\phi=4\pi/3$. }
\end{figure}

\begin{figure}

\includegraphics[width=4.5cm]{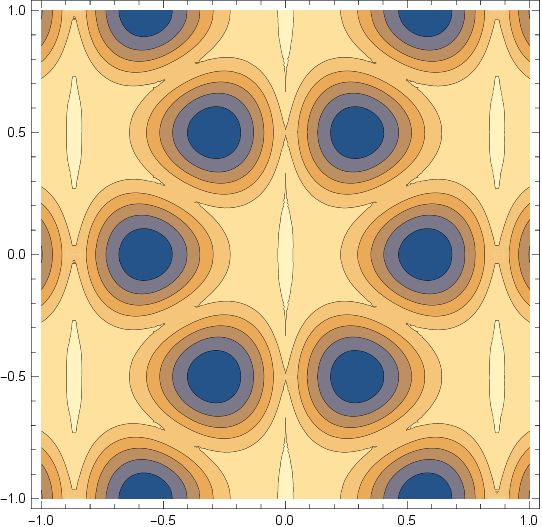}
\includegraphics[width=4.5cm]{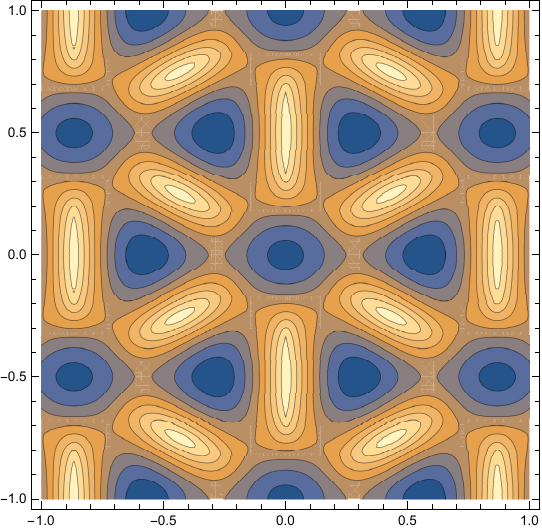}
\includegraphics[width=4.5cm]{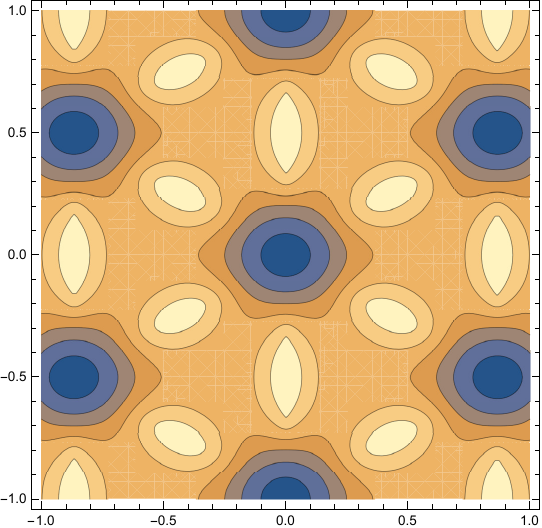}
\caption{\label{fig:2} Contour plot of $x \mapsto \lambda_-(V(x))$ with $U=0,\beta=1$ and $\phi=1.32,1.94,2.31$ from left to right exhibiting the features of a honeycomb, Kagome, and triangular lattice. }
\end{figure}

As Figures \ref{fig:1} and \ref{fig:2} shows, Assumption \ref{ass1} is not always satisfied in this model. Our work here crucially relies on this assumption and follow-up work will be devoted to studying other symmetry configurations. We do believe the band structure to change quite dramatically for these other configurations, see Figure \ref{fig:phi}. 
\section{Harmonic approximation}
\label{sec:harmony}
We consider the semiclassical TMD Hamiltonian $H=-h^2\Delta + V$ with Hermitian matrix $V \in C^{\infty}(\mathbb R^2/\Gamma; \CC^{2\times 2})$ and eigenvalues
$$\lambda_-(V(x)) <\lambda_+(V(x)),\quad x\in \mathbb R^2/\Gamma. $$ 
There exists a smooth unitary matrix $\mathcal U(x)$ such that 
\[ \mathcal U(x)^* V \mathcal U(x) = \operatorname{diag}(\lambda_+(V(x)), \lambda_-(V(x))),\]
where $\mathcal U(x)= \begin{pmatrix} v_+(x), v_-(x)\end{pmatrix}$ is a smooth orthonormal basis. 
We then find by conjugating the Hamiltonian $H$ by the same unitary
\begin{equation}
\label{eq:leading_order}
 \begin{split}
 \widetilde{H}: &= \mathcal U(x)^* H \mathcal U(x) = \mathcal U(x)^* (-h^2 \Delta ) \mathcal U(x) + \operatorname{diag}(\lambda_+(V), \lambda_-(V))\\
  &=-h^2 \Delta +  \operatorname{diag}(\lambda_+(V), \lambda_-(V)) -2h^2\sum_{i=1}^2\mathcal U(x)^* (\partial_{x_i}\mathcal U) \partial_{x_i} -h^2 \mathcal U(x)^*  (\Delta  \mathcal U)\\
  & =: -h^2 \Delta +  \operatorname{diag}(\lambda_+(V(x)), \lambda_-(V)) + R(h,x,hD_x)
  \end{split}
\end{equation}
where $R(h,x,hD_x)\in h\mathrm{Diff}_h^1(\mathbb R^2/\Gamma; \CC^{2\times2})$ is of the general form 
\[h(A_1(x)hD_{x_1} +A_2(x)hD_{x_2})+ h^2B(x)\text{ for }A_i,B \in L^{\infty}(\mathbb R^2;\mathbb C^{2\times 2}).\]
For notational convenience, we write $\lambda_{\pm}(x):= \lambda_{\pm}(V(x))$. We now show that the low lying spectrum on $L^2(\RR^2)$ of the harmonic oscillator 
\begin{equation}
    \label{eq:h-osc}
    H_{\text{osc}}:= -h^2\Delta + \operatorname{diag}(\langle x,V_+x\rangle,\langle x,V_-x\rangle) + \operatorname{diag}(\lambda_{+}(0),\lambda_{-}(0))
\end{equation}
with $V_{\pm} = \operatorname{diag}(\partial_{x_1}^2 \lambda_{\pm}(0),\partial_{x_2}^2 \lambda_{\pm}(0))$
and the corresponding spectrum of the perturbed operator
\begin{equation}
    \label{eq:h-pert}
    H_{\text{pert}}:= H_{\text{osc}} + R(h,x,hD_x) = H_{\text{osc}} + h \sum_{i=1}^2 A_i(x) hD_{x_i} + h^2 B(x)
\end{equation}
differs by $\mathcal{O}(h^{3/2})$. It is clear by general properties from semiclassical analysis that the term $h^2B$ affects the spectrum only at order $\mathcal O(h^2)$ and thus this contribution is subleading to the linear harmonic oscillator spacing and we can take $R(x) = h^2 \sum_{i=1}^2 A_i(x) D_{x_i}$. 

Let $\lambda \notin \Spec_{L^2(\RR^2)}(H_{\text{osc}}),$ then formally we have 
\[ (\lambda-H_{\text{pert}})^{-1} = (\operatorname{Id}-R(h,x,hD_x)(\lambda-H_{\text{osc}})^{-1})^{-1} (\lambda-H_{\text{osc}})^{-1}.\]
The Neumann series exists if and only if $S:=R(h,x,hD_x)(\lambda-H_{\text{osc}})^{-1}$ has operator norm $<1.$
We verify this in the following lemma
\begin{lemm}
\label{l:harmonic}
Let $\lambda_-(0)+rh$ be an eigenvalue of $H_{\text{osc}}$ with
\begin{equation}
    \label{eq:harmonic_eigen}
    r \in  \mathcal A:=\Big\{ \sum_{i=1}^2 (2m_i+1)\sqrt{\lambda_{i}(D^2\lambda_-(0))}; m\in \mathbb N_0^2\Big\}, 
\end{equation}
then for $h$ sufficiently small, depending on $r$, the eigenvalues of $H_{\text{pert}}$ remain with multiplicity inside the ball $B(\lambda_-(0)+rh,ch^{3/2})$ for a sufficiently large constant $c=c(r)$.
\end{lemm}
\begin{proof}
By the triangle inequality 
\[ \Vert S \Vert \le 2 \max_i \Vert A_i \Vert\Vert T_i\Vert, \]
where 
\[T_i := h^2D_{x_i}(\lambda-H_{\text{osc}})^{-1}.\]

Without loss of generality, we give the proof for the standard harmonic oscillator and focus on $D_{x_1}$

\[H_{\text{osc}}= \sum_{i=1}^2 \frac{p_i^2}{2} +\frac{x_i^2}{2}\]
with $p_i=hD_{x_i}.$

The spectral decomposition then implies that, for $(\vert i \rangle )_{ i \in \mathbb N_0^2}$ the harmonic oscillator basis,
\[ (\lambda-H_{\text{osc}})^{-1}= \sum_{i_1,i_2=0}^{\infty} \frac{\vert i_1,i_2 \rangle \langle i_1,i_2 \vert}{(\lambda-h(i_1+i_2+1))}.\]

The differential $hD_{x_1}$ has the representation 
\[ hD_{x_1} = i \sqrt{\frac{h}{2}}(a_1^*-a_1).\]
using standard creation and annihilation operators that satisfy with respect to the harmonic oscillator basis
\[ \begin{split}
      a_1^*|i\rangle = \sqrt{i_1 + 1} | i_1 + 1,i_2\rangle \text{ and } 
          a|i\rangle = \sqrt{i_1} | i_1 - 1,i_2\rangle.
\end{split}\]
Thus, writing $i=(i_1,i_2)$ we find 
\[h^2D_{x_i}(\lambda-H_{\text{osc}})^{-1} = h^{3/2}  \sum_{i_1,i_2=0}^{\infty}\frac{\sqrt{i_1+1}\vert i+\hat{e}_1 \rangle \langle i \vert  - \sqrt{i_1}\vert i-\hat{e}_1 \rangle \langle i \vert }{\lambda-h(i_1+i_2+1)}.\]

We have thus decomposed the operator $T_1 = T_+ -T_-$ into two summands 
\[ T_+:=h^{3/2}  \sum_{i_1,i_2=0}^{\infty}\frac{\sqrt{i_1+1}\vert i+\hat e_1 \rangle \langle i \vert}{\lambda-h(i_1+i_2+1)} \text{ and }T_-:=h^{3/2}  \sum_{i_1,i_2=0}^{\infty}\frac{\sqrt{i_1}\vert i-\hat e_1 \rangle \langle i \vert}{\lambda-h(i_1+i_2+1)}.\]

We also introduce 
\[M_-:=T_-^*T_- = h^3 \sum_{i_1,i_2=0}^{\infty} \frac{i_1 \vert i \rangle \langle i \vert}{(\lambda-h(i_1+i_2+1))^2}.\]
Therefore $M_-$ is a diagonal operator. The operator norms are 
\[ \Vert T_- \Vert^2 = \Vert M_- \Vert =\max_{i} \frac{i_1h^3}{(\lambda - h(i_1+i_2+1))^2}. \]
Now, for $r\in\mathcal{A}$ fixed, we choose $\lambda \in B(\lambda_-(0)+rh,h/2)$, then 
\[ \Vert M_- \Vert = \frac{(r-1)h^3}{\vert \lambda-rh\vert^2}. \]
We thus find that as long as $\lambda \in B(\lambda_-(0)+rh,h/2) \setminus B(\lambda_-(0)+rh,ch^{3/2})$ with $c$ a sufficiently large constant, that depends on $r$, we have 
\[\Vert M_-\Vert \ll 1.\]
Thus, doing all analogous computations for all the other operators, we find that $\Vert S \Vert <1$ which implies the claim.

In particular, this shows that 
\[(\lambda-H_{\operatorname{pert}})^{-1} = (\operatorname{Id}-\mathcal O(\sqrt{h}))^{-1} (\lambda-H_{\operatorname{osc}})^{-1}.\]
Thus, we conclude that the range of the spectral projections associated with eigenvalues inside $B(\lambda_-(0)+rh,h/2)$ have the same dimension by showing that the norm of the spectral projections is smaller than one:
\[\begin{split} &\left \Vert \frac{1}{2\pi i} \int_{\partial B(\lambda_-(0)+rh,h/2)} (\lambda-H_{\operatorname{pert}})^{-1}-(\lambda-H_{\operatorname{osc}})^{-1} d \lambda \right\Vert  \\
&\le \frac{1}{2\pi }\left \lvert \int_{\partial B(\lambda_-(0)+rh,h/2)} \mathcal O(\sqrt{h})\left \lVert (\lambda-H_{\operatorname{osc}})^{-1} \right\rVert \ d\lambda \right \rvert  \\
&\le \frac{1}{2\pi }\left \lvert \int_{\partial B(\lambda_-(0)+rh,h/2)} \mathcal O(\sqrt{h}) \mathcal O(1/h) \ d\lambda \right \rvert= \mathcal O(\sqrt{h}) \ll 1,
\end{split}\]
which implies the spectral stability.
\end{proof}

Since the TMD Hamiltonian is periodic, recall that in Section \ref{sec:Hamiltonian}, we have 
\[ \Spec_{L^2(\mathbb R^2)}(H) = \bigcup_{\mathbf k \in \R^2/\Gamma^*} \Spec_{L^2(\mathbb R^2/\Gamma)}(H_{\mathbf k})\]
with 
\begin{equation}
\label{eq:HK}
H_{\mathbf k} = h^2(D_x+\mathbf k)^2 +V.
\end{equation}
Note that eigenvalues of the semiclassical principal symbol $\sigma_h(H_{\bf k})$ are given by $ \lambda_{-,\mathbf k}(x,\xi) < \lambda_{+,\mathbf k}(x,\xi) $ with 
\[\begin{split} \lambda_{\pm,\mathbf k}(x,\xi) &=  \vert \xi+\mathbf k\vert^2 + \lambda_{\pm}(V(x)),
\end{split}\]
where $\lambda_{\pm}(V(x))$ is defined in \eqref{eq:eigenvalue}, are gapped with associated (real-analytic) eigenvectors 
\[\mathcal U(x) = (v_+(x),v_-(x)), \quad v_{\pm}(x) = \Big(\frac{V_{\uparrow}( x)  - V_{\downarrow}(  x) +U}{2} \pm U_{\text{eff}}( x),\beta T( x)^*\Big)^t.\]

We can then use a cut-off function $\chi \in C_c^{\infty}(\RR^2)$ with $\chi(x)=1$ for $\vert x \vert \le \delta$ and $\chi(x)=0$ for $\vert x \vert \ge 2\delta$ to define a new single-well operator
\begin{equation}
    \label{eq:single-well}
\tilde H_{\text{well}}:= \tilde H + (1-\chi(x)), \quad \tilde H = \mathcal U(x)^* H(x) \mathcal U(x).
\end{equation}
We study the spectrum of this operator using the harmonic approximation of the potential in the following subsection.

 \begin{figure}
 \includegraphics[width=7cm]{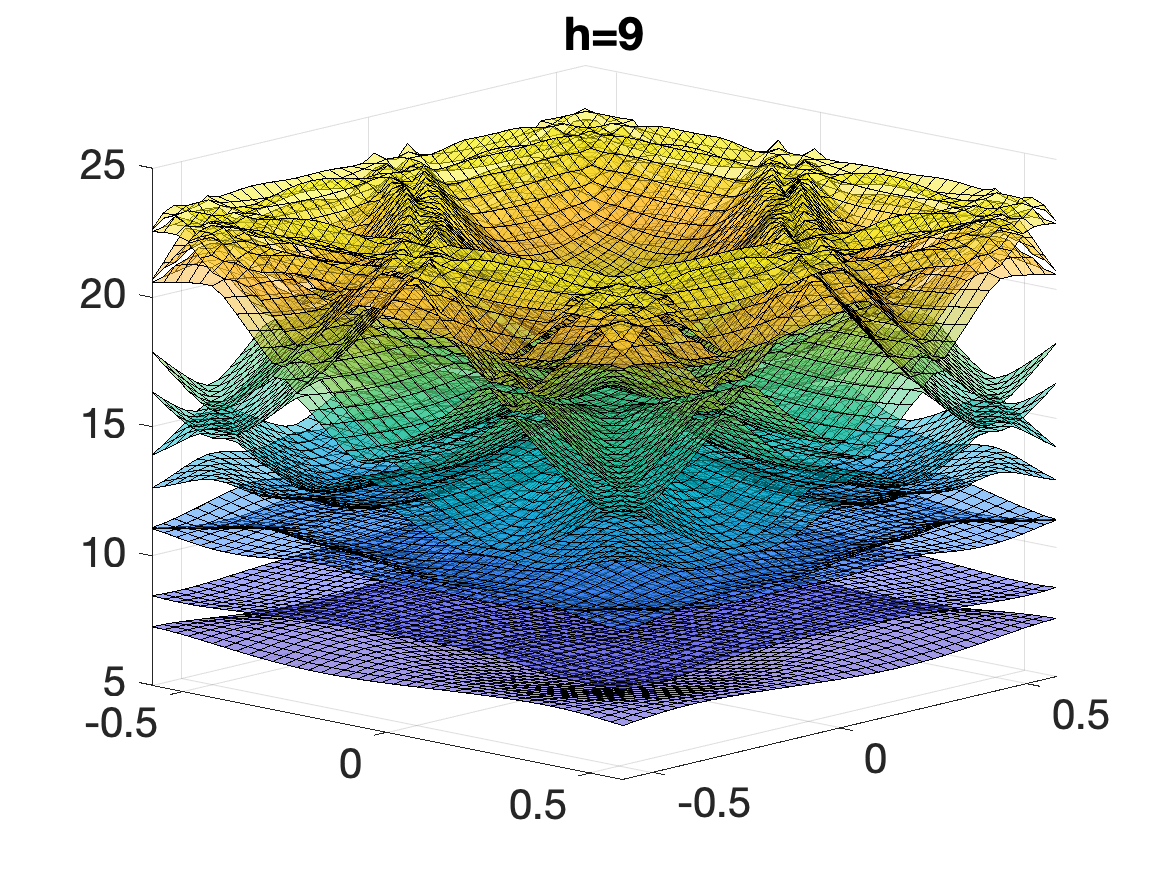} \includegraphics[width=7cm]{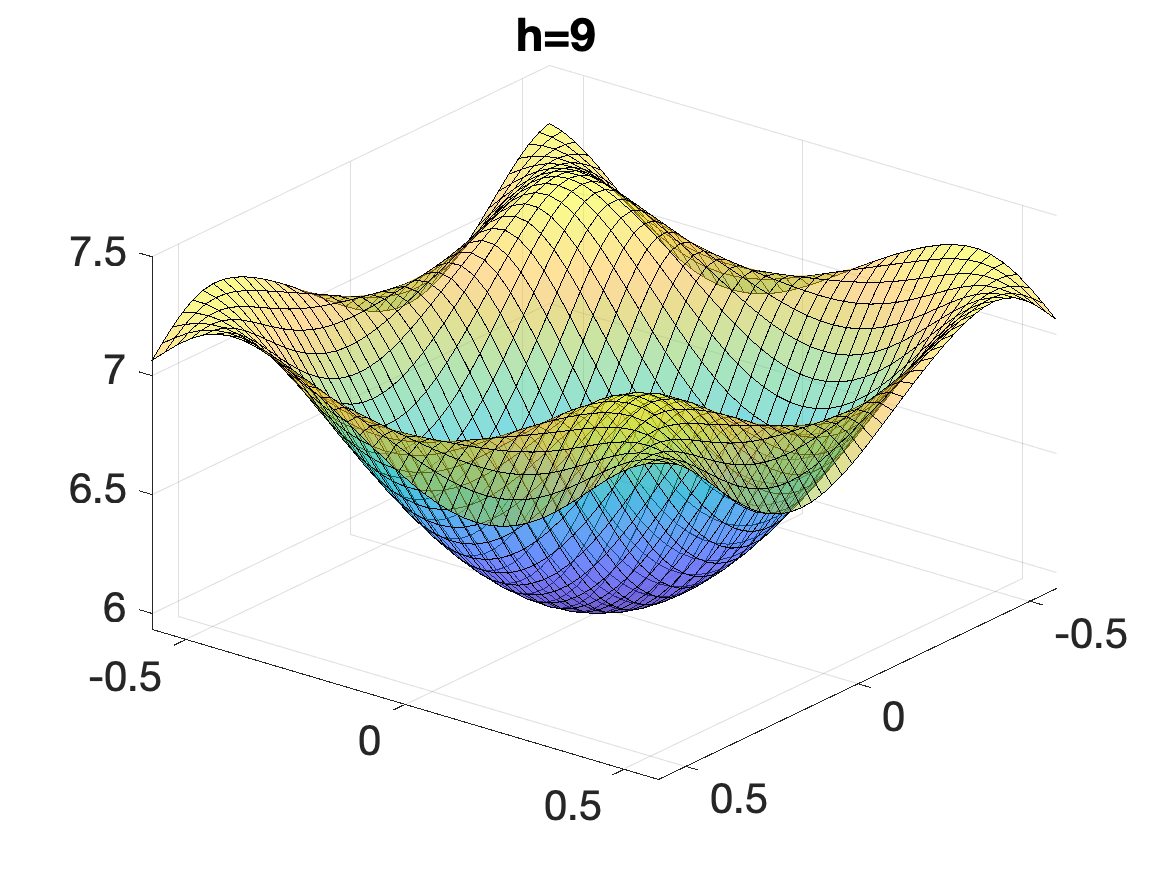}
  \includegraphics[width=7cm]{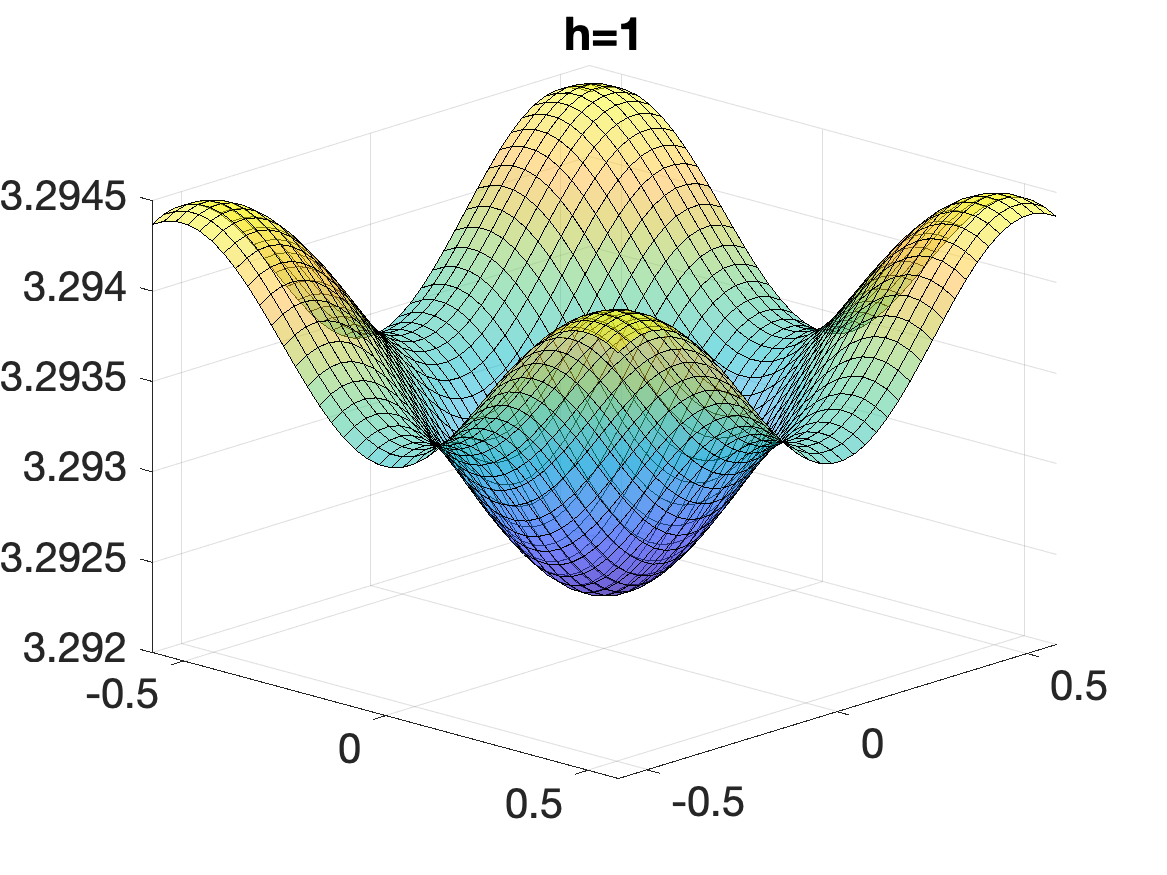} \includegraphics[width=7cm]{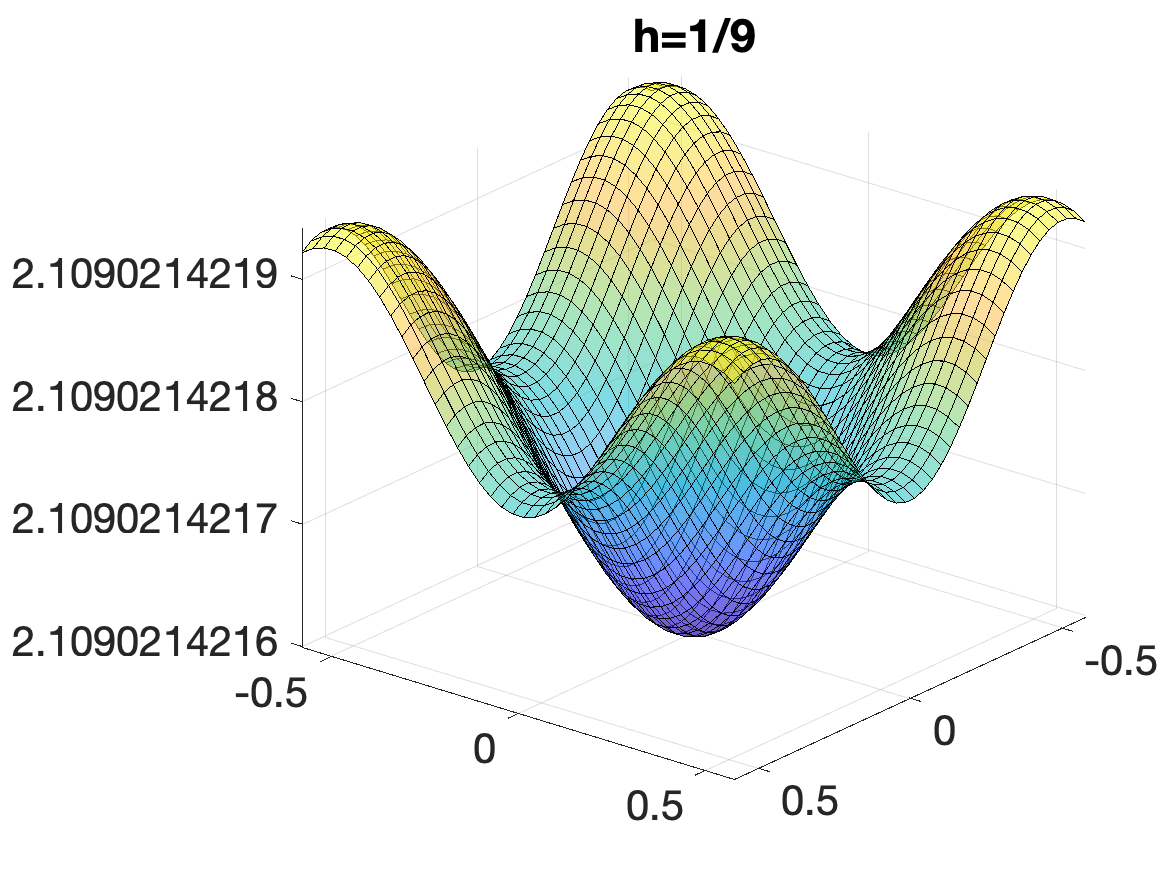}
  \caption{Different $h$, $\alpha=\beta=1$. All lowest bands for $h=9$ and lowest band for different parameters of $h$ showing exponential flattening.}
 \end{figure}
 
  \begin{figure}
 \includegraphics[width=7cm]{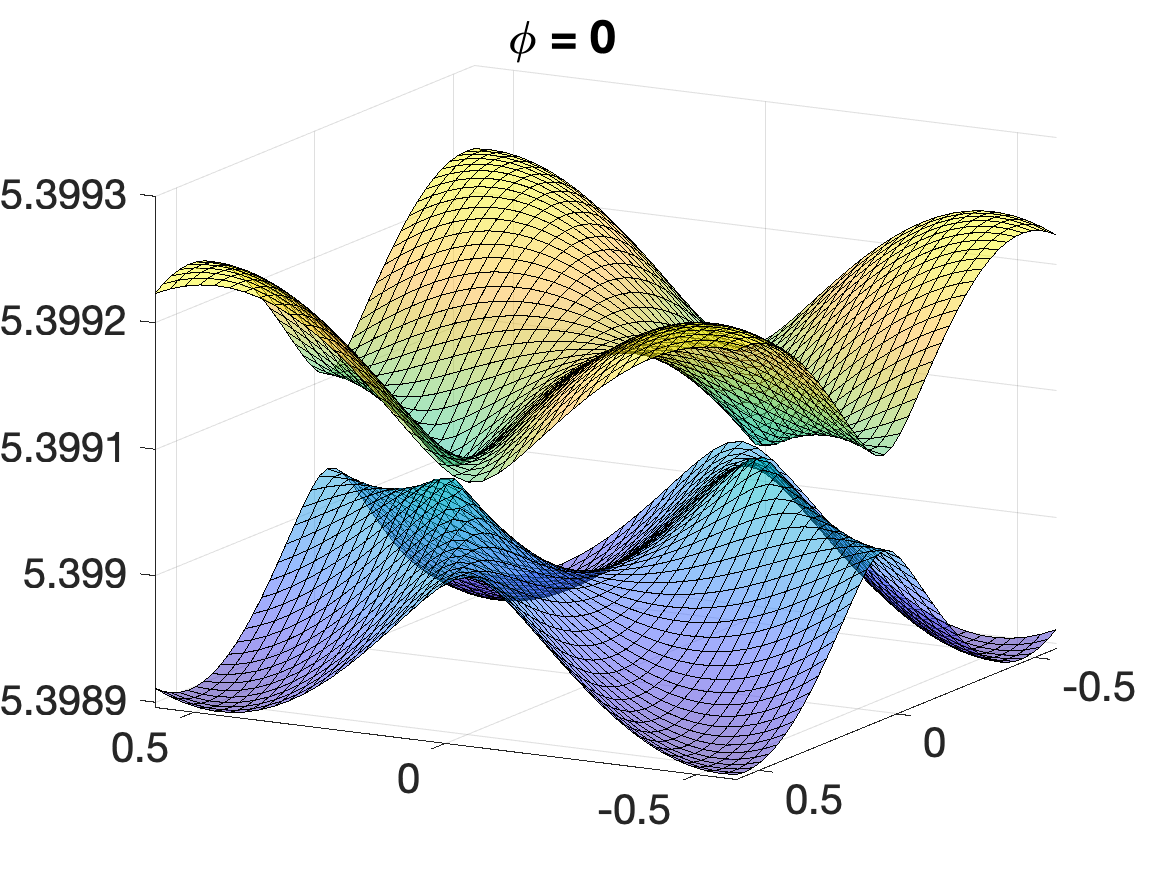} \includegraphics[width=7cm]{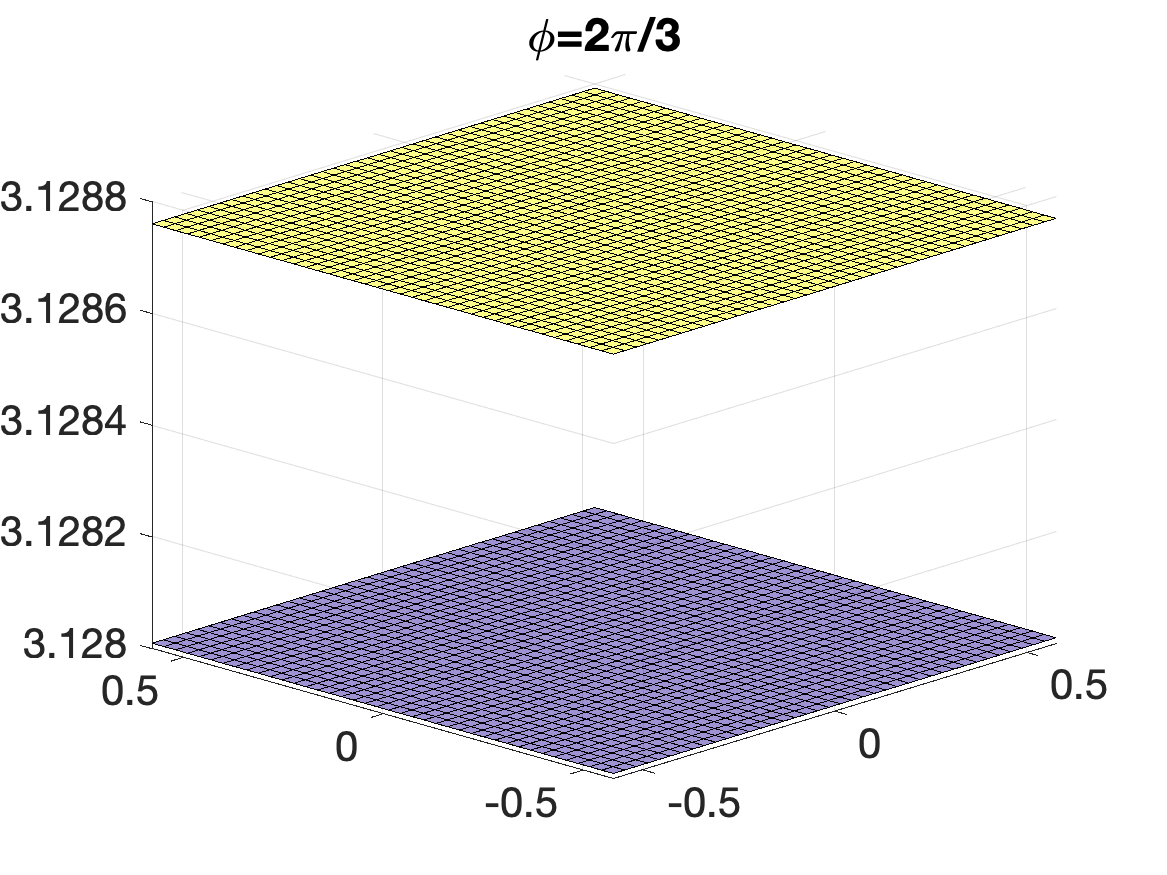}
  \caption{  \label{fig:phi} Lowest two bands for $h=1/9$, $\alpha=1$ and $\beta=1/9$ with $\phi=0$ (left) almost exhibiting Dirac points and $\phi = 2\pi/3$ (right).}
 \end{figure}

\subsection{Harmonic approximation with remainder terms}
We want to show that the low lying spectrum of $\tilde{H}_{\text{well}}: H^2(\RR^2)\to L^2(\RR^2)$ as well as $\tilde{H}_{\bf k}:  H^2(\RR^2/\Gamma)\to L^2(\RR^2/\Gamma)$ defined in equation \eqref{eq:leading_order} can be approximated by the spectrum of $H_{\text{pert}}$ defined in \eqref{eq:h-pert}. 
We can therefore approximate the spectrum of $\tilde{H}_{\text{well}}$ and $\tilde{H}_{\mathbf{k}}$ by the spectrum of harmonic oscillator $H_{\text{osc}}$ using Lemma \ref{l:harmonic}.
This idea leads to the proof of Theorem \ref{theo:harmonic_approximation}.

\begin{proof}[Proof of Theo.\,\ref{theo:harmonic_approximation}]

We start by proving the first part and then sketch the modifications for the second part at the end of the proof. 
We approximate the potential $\mathrm{diag}(\lambda_{+}(V), \lambda_{-}(V))$ in the Hamiltonian $\widetilde{H}$ using harmonic oscillator potentials near $x=0$. Note that by equation \eqref{eq:potential}, the minimum of the lowest eigenvalue $\lambda_{\pm}(x)$ is attained at $\lambda_{\pm}(0)=6 \cos(\phi) \pm \frac{ \sqrt{9 \beta^2 + U^2} }{2}.$
Let $V_{\pm} = \operatorname{diag}(\lambda_{1,\pm},\lambda_{2,\pm})$ where 
\[
\begin{split}\lambda_{1,\pm}&=\frac{\partial^2 \lambda_{\pm}}{\partial x_1^2}(0)=-8\pi^2  \cos(\phi) \mp
 \frac{2 }{3} \pi^2  \frac{\beta ^2}{\sqrt{9\beta^2 + U^2}},  \\
 \lambda_{2,\pm}&=\frac{\partial^2 \lambda_{\pm}}{\partial x_2^2}(0)=-8\pi^2 \cos(\phi) \mp
  \pi^2  \frac{6\beta ^2  }{\sqrt{9\beta^2 + U^2}} .\end{split}
\]
Therefore, we consider the reference operator that is harmonic with the same remainder term as in equation \eqref{eq:leading_order}
\[K_{R}(x,D_x) :=- h^2\Delta + \operatorname{diag}(\langle x,V_+x\rangle,\langle x,V_-x\rangle) + \operatorname{diag}(\lambda_{+}(0),\lambda_{-}(0)) + R(h,x,hD_x).\]
Let $J(x):=\rho(h^{-2/5}x)$ with $\rho \in C_c^{\infty}(\mathbb R^2)$ with $0\le \rho \le 1$ where $\rho(x)=1$ for $\vert x\vert \le 1$ and $\rho(x)=0$ for $\vert x \vert \ge 2.$ 
Recall that 
$$\tilde{H}_{\text{well}} = \operatorname{diag}(\lambda_{+}(x,hD_x), \lambda_{-}(x,hD_x)) + R(h,x,hD_x) + (1-\chi).$$ 
We then have
\begin{equation}
\label{eq:diff1}
 \Vert J (\tilde{H}_{\text{well}} -K_{R} )J \Vert = \mathcal O(h^{6/5}).
 \end{equation}
which follows from the harmonic approximation (cf.~\cite[Equation (11.5)]{CFKS}) as the remainder terms cancel off.

Note that the eigenvectors associated with the low-lying spectrum of the operator $K_{0}$ are for $n \in \mathbb N_0^2$ of the form 
\[ \varphi_{n}(x) = \frac{\left(\frac{\sqrt{\lambda_1\lambda_2}}{h^2\pi^2 }\right)^{1/4}e^{
- \Big(\frac{\sqrt{\lambda_1} x_1^2+ \sqrt{\lambda_2}x_2^2}{2h}\Big)}}{2^{n_1+n_2}\,(n_1)!(n_2)!}     H_{n_1}\left(\sqrt{\frac{\lambda_1}{h}} x_1 \right)H_{n_2}\left(\sqrt{\frac{\lambda_2}{h}} x_2 \right) \hat{e}_2,\]
where $\hat{e}_2 =(0,1)^t$ and $H_n$ the $n$-th Hermite polynomial.

We also define $J_0:=\sqrt{1-J(x)^2}$ and $\psi_{n} := J\varphi_{n}$. We then have
\begin{gather}
    \label{eq:aux-1}
    \langle \psi_{n}, \psi_{m}  \rangle = \delta_{m,n} + \mathcal{O}(e^{-C/h^{1/5}})\\
    \label{eq:aux-2}
    \langle \varphi_{n}, (D_xJ_a(h))^2 \varphi_{n}  \rangle = \mathcal{O}(e^{-C/h^{1/5}}),\ J_a\in\{J_0,J\}.
\end{gather}
By equation \eqref{eq:diff1}, we have
\begin{equation}
    \begin{split}
    \label{eq:aux-5}
        \langle \psi_{n}, \tilde{H}_{\text{well}}\psi_{m}  \rangle 
        & = \langle J \varphi_{n}, \tilde{H}_{\text{well}}(J \varphi_{m}) \rangle \\
        & = \langle J \varphi_{n}, K_{R} (J \varphi_{m}) \rangle + \mathcal{O}(h^{6/5}).
    \end{split}
\end{equation}
Using equation \eqref{eq:aux-1} and \eqref{eq:aux-2}, we obtain that
\begin{equation}
\label{eq:aux-3}
    \begin{split}
        \langle J \varphi_{n}, K_{R} (J(h) \varphi_{m}) \rangle 
        & = \langle J \varphi_{n}, K_{0} (J \varphi_{m}) \rangle + \langle J \varphi_{n}, R(h,x,hD_x) (J \varphi_{m}) \rangle \\
        & = (h e_n+\lambda_-(0))\delta_{m,n} + \mathcal{O}(h^{3/2}),
    \end{split}
\end{equation}
where $h e_n+\lambda_-(0)$ is the corresponding eigenvalue of $\varphi_n(x)$ with $e_n\in\mathcal{A}$ defined in \eqref{eq:harmonic_eigen}. Combining \eqref{eq:aux-3} with equation \eqref{eq:aux-5} and using Lemma \ref{l:harmonic}, we obtain
\begin{equation}
    \label{eq:aux-6}
    \langle \psi_{n}, \tilde{H}_{\text{well}}\psi_{m}  \rangle = E'_n\delta_{m,n} + \mathcal{O}(h^{6/5}).
\end{equation}
Using Rayleigh--Ritz principle we have the upper bound for $E_{n}(h)$:
\begin{equation}
    \label{eq:ub}
    \limsup_{h\downarrow 0} \frac{E_n - E'_n}{h} \leq 0.
\end{equation}

Now we look for a lower bound for eigenvalues $E_{n}(h)$ of $\tilde{H}_{\text{well}}$. Using the IMS formula
\begin{equation}
\label{eq:IMS-formula}
\begin{split}
\tilde{H}_{\text{well}} = J_0 \tilde{H}_{\text{well}} J_0 + J \tilde{H}_{\text{well}} J - h^2 (\vert D_xJ_0\vert^2 + \vert D_x J \vert^2)
\end{split} 
\end{equation}
and equation \eqref{eq:diff1}, we have 
\begin{equation}
\label{eq:IMS}
\begin{split}
\tilde{H}_{\text{well}}
&=J_0 \tilde{H}_{\text{well}} J_0 +J\tilde{H}_{\text{well}}J+ \mathcal O(h^{6/5})\\
&=J_0\tilde{H}_{\text{well}} J_0 + JK_{R} J + \mathcal O(h^{6/5}),
\end{split} 
\end{equation}
whereas the IMS formula \eqref{eq:IMS-formula} follows from summing over the double commutators
\[  - 2 h^2 \vert D_x J_a \vert^2 = [ J_a,[J_a, \tilde{H}_{\text{well}}]] = J_a^2 \tilde{H}_{\text{well}} + \tilde{H}_{\text{well}} J_a^2 - 2 J_a \tilde{H}_{\text{well}} J_a, \ J_a \in \{J_0,J\}.  \]
For the first term on the RHS of the equation \eqref{eq:IMS}, by the support property of $J_0$, we have $J_0 \langle x,V_{\pm }x\rangle J_0 \geq Ch^{4/5}J_0^2$, where we use $A\geq B$ to denote that $A-B$ is a positive operator. 
Recall that we have $E'_n(h) = \lambda_-(0)+ h e_n +\mathcal{O}(h^{3/2})$ by Lemma \ref{l:harmonic}. Thus, for $h>0$ small enough and $\lambda_-(0)+ he = E' \in (E'_{n}(h), E'_{n+1}(h))$ with $n$ arbitrary, we have
\[J_0\tilde{H}_{\text{well}}J_0 \geq C h^{4/5} J_0^2 \geq h e J_0^2 = E' J_0^2 ,\]
where we use the fact that $-h^2\Delta$ is a positive operator, 
and that the remainder term $R(x)=\mathcal{O}(h)$ can be absorbed into the main term. 
For the second term on the RHS of the equation \eqref{eq:IMS}, as $E'_n<E'<E'_{n+1}$, we have that
\begin{equation}
    \label{eq:aux-4}
    J K_{ R}J \ge J F_n J + E' J^2,
\end{equation}
where $F_n$ is a rank $n$ operator. Thus, using equation \eqref{eq:IMS} and \eqref{eq:aux-4} we obtain
\begin{equation}
\label{eq:IMS1}
\begin{split}
\tilde{H}_{\text{well}}
&= J_0 \tilde{H}_{\text{well}} J_0 + JK_{R} J + \mathcal O(h^{6/5}) \\ 
&\ge E' (J_0^2+J^2) +  J (K_{R}-E') J +  \mathcal O(h^{6/5})\\
& = E' +  J F_n J +  \mathcal O(h^{6/5}).
\end{split} 
\end{equation}
As $E'$ can be taken arbitrarily close to $E_{n+1}$, equation \eqref{eq:IMS1} yields that the lower bound for the eigenvalue $E_{n+1}(h)$ of $\tilde{H}_{\text{well}}$:
\begin{equation}
\label{eq:lb}
\liminf_{h\downarrow 0} \frac{E_{n+1} - E'_{n+1}}{h} \geq 0.
\end{equation}
The first part of the theorem follows from combining equation \eqref{eq:ub} and \eqref{eq:lb}.

For the second part, we define the periodized versions of $J$ and $J_0$: $j(x):=\sum_{\gamma \in \Gamma} J(x+\gamma)$ and $j_0(x):=\sqrt{1-j(x)^2}.$

One then finds as in \eqref{eq:diff1} that 
\[ \Vert j \tilde H_{\mathbf k} j - k_R(x,h(D_x+\mathbf k)) \Vert_{L^2(\mathbb R^2/\Gamma)} = \mathcal O(h^{6/5}),\]
where we defined
$k_R(x,h(D_x+\mathbf k)):=\sum_{\gamma \in \Gamma} J(x+\gamma) K_R(x+\gamma,h(D_x+\mathbf k)) J(x+\gamma).$

In the periodic case, one then uses periodic Bloch functions as trial functions replacing $\varphi_n$ and $\psi_n$ given by
\[\begin{split}
\Phi_n(\mathbf k,x)&=\sum_{\gamma \in \Gamma} e^{-i\mathbf k(\gamma+x)}\varphi_n(x+\gamma) \text{ and }\\    
\Psi_n(\mathbf k,x)&=\sum_{\gamma \in \Gamma} e^{-i\mathbf k(\gamma + x)}J(x+\gamma)\varphi_n(x+\gamma) .
\end{split}\]
Then one readily shows the analogue of \eqref{eq:ub}. To get also the opposite bound \eqref{eq:lb} one uses the analogue of the periodic IMS formula \eqref{eq:IMS-formula}
\[ \tilde H_{\mathbf k} = j_0 \tilde H_{\mathbf k} j_0 + j \tilde H_{\mathbf k} j -h^2( \vert D_xj_0\vert^2 + \vert D_xj \vert^2 ). \]
Thus, we have as in \eqref{eq:IMS}
\[\tilde H_{\mathbf k} = j_0 \tilde H_{\mathbf k} j_0 +k_{R}(x,h(D_x+\mathbf k)) + \mathcal O(h^{6/5}).\]
Since $JF_n J$ is a self-adjoint operator of rank at most $n$, we can write it by the spectral theorem as 
$$JF_nJ = \sum_{i=1}^n \lambda_i \vert \chi_i \rangle \langle \chi_i  \vert.$$
This implies that $(JF_nJ)(\bullet+\gamma) = \sum_{i=1}^n \lambda_i \vert \chi_i (\bullet+\gamma) \rangle \langle \chi_i (\bullet+\gamma) \vert$ and therefore the analogue of \eqref{eq:aux-4} 
\[k_R \ge \sum_{i=1}^n \lambda_ i \sum_{\gamma \in \Gamma} \vert \chi_i (\bullet+\gamma) \rangle \langle \chi_i (\bullet+\gamma) \vert + E' j^2. \]
This allows us to show \eqref{eq:lb} following the same steps as in the first part of the proof.
\end{proof}

\section{Topologically trivial bands}
\label{sec:topology}
The goal of this section is to prove Theorem \ref{theo:triv_bands}.
We shall introduce some preliminary concepts
\begin{defi}
\label{defi:distance}
Let $E,F$ be two closed subspaces of a Hilbert space $H$ with orthogonal projections $\Pi_{E}$ and $\Pi_{F}$ respectively, we define the non-symmetric distance function
\[ d(E,F):=\sup_{x \in E; \Vert x \Vert=1} \inf_{y \in F} \Vert x-y\Vert = \Vert (\operatorname{id}-\Pi_F)\vert_E \Vert = \Vert \Pi_E-\Pi_F \Pi_E \Vert =\Vert \Pi_E-\Pi_E \Pi_F \Vert. \]
\end{defi}
One then has the following
\begin{lemm}{\cite[Prop.1.4]{HS84}}
Let $E,F$ be two closed subspaces with $d(E,F)<1$ and $d(F,E)<1$, then $d(F,E)=d(E,F)$ and $\Pi_{E}\vert_{F}$ and $\Pi_{F}\vert_{E}$ are continuously invertible operators.
\end{lemm}
Let $A$ be a self-adjoint operator on a Hilbert space $H$. Let $I$ be a compact interval, then we have the following result to estimate $d(E,F)$ where $F:=\operatorname{ran}(\indic_{I}(A))$ and $E$ is spanned by trial functions. 
\begin{lemm}{\cite[Prop.2.5]{HS84}}
\label{lemm:auxiliary}
Let $\psi_1,...,\psi_N$ be linearly independent elements spanning a space $E:=\operatorname{span}\{\psi_j; j \in \{1,..,N\}\}$ of $H$ and numbers $\mu_1,...,\mu_N \in I$ such that 
\[ r_j:=(A-\mu_j)\psi_j \text{ with } \Vert r_j \Vert \le \varepsilon.\]
Let $a>0$ such that $\Spec(A)\cap (I+B(0,2a) \setminus I) = \emptyset.$
Then, we have
\[ d(E,F) \le \frac{N^{1/2} \varepsilon}{a \vert \lambda_s^{\text{min}}\vert^{1/2}}\]
where $\lambda_s^{\text{min}}$ is the smallest eigenvalue of the gramian matrix $S = (\langle \psi_j,\psi_k \rangle ).$
\end{lemm}
We shall now give the proof of Theorem \ref{theo:triv_bands}.
\begin{proof}[Proof of Theo.\,\ref{theo:triv_bands}]
We first observe that by the $\mathbf k$ independent unitary equivalence
\[ H_{\mathbf k} = \mathcal U \tilde H_{\mathbf k} \mathcal U^*,\]
it suffices to compute the Chern number of the bands of $\tilde H_{\mathbf k}.$

We consider the scalar operator Schr\"odinger operator
\[ \lambda_{\pm}(x,hD_x) = -h^2 \Delta + \lambda_{\pm}(V(x)).\]
Since its potential $\lambda_{\pm}(V(x))=\frac{V_{\uparrow}(x)+ V_{\downarrow}(x)} {2} \pm U_{\text{eff}}(x)$ is real-valued, it follows that $\lambda_{\pm}(x,hD_x)$ satisfies time-reversal symmetry:

Let $Cu(x)=\overline{u(x)}$ with $C^2 = I$, then
\[ C{\lambda_{\pm}(x,hD_x) u } = \lambda_{\pm}(x,hD_x) C{u}.\]
We have Bloch Hamiltonians 
\[\begin{split} 
\lambda_{\pm}(x,hD_x+\mathbf k)&: H^2(\RR^2/\Gamma) \subset L^2(\RR^2/\Gamma) \to L^2(\RR^2/\Gamma) \\
\lambda_{\pm}(x,hD_x+\mathbf k) &= h^2 (D_x+\mathbf k)^2 + V_{\pm}(x).  
\end{split}\]

Let $P(\mathbf k)$ be the family of orthogonal projectors associated with a finite number of isolated bands of the Bloch-Hamiltonian
$\hat{H}_{\mathbf k}:=\operatorname{diag}(\lambda_{+}(x,D_x+\mathbf k),\lambda_{-}(x,D_x+\mathbf k))$ satisfying 
$C \hat{H}_{\mathbf k} = \hat{H}_{-\mathbf k}C. $

Then, we have by the functional calculus $CP(\mathbf k)C = P(-\mathbf k)$ such that $\partial_i P(-\mathbf k) = -C\partial_i P(\mathbf k)C$ and the Berry curvature $\Omega(\mathbf k)=K(\mathbf k) dk_1 \wedge dk_2$ satisfies
\[ \begin{split} K(-\mathbf k) &= \tr(CP(\mathbf k)C C [ \partial_i P(\mathbf k), \partial_j P(\mathbf k)]C) \\ 
&= \tr(CP(\mathbf k)C C [ \partial_i P(\mathbf k), \partial_j P(\mathbf k)]C) \\
&= \tr(CP(\mathbf k) [ \partial_i P(\mathbf k), \partial_j P(\mathbf k)]C) \\
&= -\tr(P(\mathbf k) [ \partial_i P(\mathbf k), \partial_j P(\mathbf k)]) = -K(\mathbf k),
\end{split}  \]
since $\tr(CAC)=\tr(A^*)$.
This implies the vanishing of the Chern number 
\[ \operatorname{Ch}(P)=\frac{i}{2\pi} \int_{\mathbb R^2/\Gamma^*} \Omega=0.\]

Let $\hat P(\mathbf k)$ be the spectral projection associated with $\hat{H}_{\mathbf k}$ and $\tilde P$ is the spectral projection associated with $\tilde H_{\mathbf k}$, then if we can show that
\[ \left\lVert \int^{\oplus}_{\mathbb R^2/\Gamma^*} (\hat P(\mathbf k) - \tilde P(\mathbf k)) \ d\mathbf k \right\rVert = \operatorname{ess-sup}_{\mathbf k \in \mathbb R^2/\Gamma^*} \Vert \hat P(\mathbf k)-\tilde P(\mathbf k)\Vert<1,\]
the two Chern numbers coincide by homotopy invariance \cite{B86}.

Using the distance function in Def.~\ref{defi:distance} we have for $E,F$ subspaces with projections $\Pi_E,\Pi_F$
\begin{equation}
\label{eq:distEF}
    \Vert \Pi_E-\Pi_F\Vert \le \Vert \Pi_E-\Pi_E \Pi_F \Vert + \Vert \Pi_E \Pi_F-\Pi_F \Vert = d(E,F) + d(F, E).
\end{equation}

Now, let $\hat E(\mathbf k)$ be a low-lying eigenvalue of $\hat{H}_{\mathbf k}$ with eigenfunction $ \hat{\varphi}(\mathbf k),$ then we find 
\[ \begin{split}\Vert h(D_x+\mathbf k)  \hat{\varphi}(\mathbf k)\Vert^2 &= \langle \operatorname{diag}(\hat E(\mathbf k)-V_+,\hat E(\mathbf k)-V_-) \hat{\varphi}(\mathbf k),\hat{\varphi}(\mathbf k) \rangle  \\
&\le \langle \operatorname{diag}((\hat E(\mathbf k)-V_+)_+,(\hat E(\mathbf k)-V_-)_+) \hat{\varphi}(\mathbf k),\hat{\varphi}(\mathbf k) \rangle.
\end{split}\]
Since we are interested in low-lying eigenvalues $\hat E(\mathbf k),$ the harmonic approximation shows that $\operatorname{diag}((\hat E(\mathbf k)-V_+)_+,(\hat E(\mathbf k)-V_-)_+) =\mathcal O(h)$
such that 
\[ \begin{split}\Vert h(D_x+\mathbf k) \hat{\varphi}(\mathbf k)\Vert = \mathcal O(\sqrt{h}).
\end{split}\]
For the full Hamiltonian $H_{\mathbf k}$ with eigenvalue $E(\mathbf k)$ and eigenfunction $\varphi(\mathbf k)$ the same computation shows by  denoting by $A(\mathbf k)$ the set where $E(\mathbf k)-V$ has no negative eigenvalues 
\[ \begin{split}\Vert h(D_x+\mathbf k) \varphi(\mathbf k)\Vert^2 &= \langle (E(\mathbf k)-V) \varphi(\mathbf k),\varphi(\mathbf k) \rangle  \\
&\le \langle (E(\mathbf k)-V) \chi_{A(\mathbf k)} \varphi(\mathbf k),\varphi(\mathbf k) \rangle.
\end{split}\]
Arguing as above, we also find 
\[ \begin{split}\Vert h(D_x+\mathbf k) \varphi(\mathbf k)\Vert = \mathcal O(\sqrt{h}).
\end{split}\]
The Bloch function associated with $\tilde H_{\mathbf k}$ is then given by $\tilde \varphi(\mathbf k)=\mathcal U\varphi(\mathbf k)$ and we find 
\[ h(D_x+\mathbf k) \tilde \varphi(\mathbf k) = (hD_x\mathcal U)\tilde \varphi + \mathcal U h (D_x+\mathbf k) \varphi(\mathbf k).  \]
Taking norms, we find 
\[ \Vert h(D_x+\mathbf k) \tilde \varphi(\mathbf k)\Vert = \mathcal O(\sqrt{h}).  \]
We have thus shown that 
\[ \Vert (\hat{H}_{\mathbf k}-\hat{E}(\mathbf k)) \tilde \varphi(\mathbf k) \Vert = \mathcal O(h^{3/2}) \quad \text{and} \quad \Vert (\tilde{H}_{\mathbf k}-E(\mathbf k)) \hat \varphi(\mathbf k) \Vert = \mathcal O(h^{3/2}). \]

Taking $E = \operatorname{span}\{\hat \varphi(\mathbf k)\}$ and $F =\operatorname{span}\{\tilde \varphi(\mathbf k)\}$ in \eqref{eq:distEF} and applying Lemma \ref{lemm:auxiliary} to both operators $\hat{H}_{\mathbf k}$ and $\tilde {H}_{\mathbf k}$  with $I=(E(\mathbf k) - c_1 h^{3/2},E(\mathbf k)+ c_1h^{3/2})$ for $c_1$ sufficiently large and $a=c_2 h$ with $c_2$ sufficiently small, then Lemma \ref{lemm:auxiliary} (note that all conditions are satisfied by Theorem \ref{theo:harmonic_approximation}) implies for $N=1$, $\varepsilon=\mathcal O(h^{3/2})$, $\lambda_s^{\text{min}}=1$
\[ d(E,F)+d(F,E) = \mathcal O(\sqrt{h}) \]
which finishes the argument.
\end{proof}

\section{Agmon estimate for matrix-valued potentials}
\label{sec:Agmon}
We start by introducing the Agmon distance for Schr\"odinger operators with matrix-valued potentials and follow the presentation in \cite[Chapter 3]{HS12}. It turns out that the adaptation to matrix-valued operators is straightforward if one replaces the potential matrix by its lowest eigenvalue.
We assume that $V$ is a continuous matrix-valued potential. 
Then the eigenvalues $\lambda_1(V) \le ...\le \lambda_n(V)$ also depend continuously on $V$, see \cite[Ch.~2, Sec.~5, Theo.~5.1]{kato}. Let $E$ be such that $\supp(E-\lambda_1(V))_+$ is compact. 
\begin{defi}
Let $x \in \mathbb R^n$ and let $\xi,\eta \in T_x(\mathbb R^n)$. The Agmon (degenerate) inner product is then defined by
\[ \langle \xi,\eta \rangle_x = (\lambda_1(V(x))-E)_+ \langle \xi,\eta \rangle_{\mathbb R^n}.\]
Let $\gamma: [0,1] \to \mathbb R^n$ be an absolutely continuous curve, the Agmon length is  
\[ L_E(\gamma):=\int_0^1 (\lambda_1(V(\gamma(t)))-E)_+^{1/2} \Vert \gamma'(t)\Vert \ dt. \]
For $x,y \in \mathbb R^n$, the Agmon distance is then defined by 
\[ d_E(x,y):=\inf_{\gamma \in P_{x,y}} L_E(\gamma).\]
Here, $P_{x,y}:=\{\gamma \in \operatorname{AC}([0,1]; \RR^n); \gamma(0)=x,\gamma(1)=y \},$ where $\operatorname{AC}([0,1]; \RR^n)$ is the space of absolutely continuous curves taking values in $\mathbb R^n.$
Since the function $d_E$ is locally Lipschitz continuous in both of its arguments, it is differentiable almost everywhere. In particular, at the points where the function is differentiable 
\[\vert \nabla_x d_E(x,y)\vert^2 \le (\lambda_1(V(x))-E)_+. \]
We also define $\rho_E(x):=d_E(x,0)$ and $\rho_E^{\delta}(x):=(1-\delta)\rho_E$ for $\delta>0$.
\end{defi}
This implies that
$\lambda_1(V(x))-E-\vert \nabla \rho_E(x)\vert^2 \ge 0$  on the set $\{x\in \RR^n: \lambda_1(V)\geq E\}$ and thus for $\delta>0$ we have 
\begin{equation}
\label{eq:estimate}
    \begin{split}
        \lambda_1(V(x))-E-\vert \nabla \rho^{\delta}_E(x)\vert^2 &\ge (1-(1-\delta)^2) (\lambda_1(V(x))-E) \\
        &=(2\delta-\delta^2) (\lambda_1(V(x))-E).
            \end{split}
\end{equation}
Then one shows the following proposition which has been originally proven in \cite[Theo.~1.1]{HS84}. We shall state it again for the matrix-valued case.  
\begin{prop}{}
\label{prop:Agmon}
Let $\Omega \subset \mathbb R^m$ be a bounded domain with $C^2$-boundary and $H_h=-h^2\Delta + V$ with $V \in C(\overline{\Omega}; \CC^{n \times n})$ self-adjoint on $L^2(\Omega; \CC^n)$ with Dirichlet boundary conditions. Then for $F_+,F_- \ge 0$ such that 
\[F_+^2 - F_-^2 = \lambda_1( V) -\vert \nabla \rho_E^{\delta} \vert^2  -E ,\quad F=F_+ + F_-\]
and $\delta>0$ we have for $u \in H^2(\Omega) \cap H^1_0(\Omega)$ and $E'\le E$
\[h^2 \Vert \nabla(e^{\rho_E^{\delta}/h}u)\Vert^2 + \frac{1}{2} \Vert e^{\rho_E^{\delta}/h} \vert F_+ u \vert \Vert^2\le \Vert \frac{1}{F} e^{\rho_E^{\delta}/h} (H_h-E') u \Vert^2  +\frac{3}{2} \Vert F_-  e^{\rho_E^{\delta}/h} u \Vert^2.
\]
\end{prop}
\begin{proof}
Integrating by parts yields
\begin{equation}
\label{eq:equations}
 \begin{split} 
&h^2 \int_{\Omega} \vert \nabla(e^{\rho_E^{\delta}/h}u)\vert^2 \ dx= -h^2 \int_{\Omega} \Delta(e^{\rho_E^{\delta}/h}u) \cdot e^{\rho_E^{\delta}/h} u \ dx  \\
=&-h^2 \int_{\Omega} e^{2\rho_E^{\delta}/h} \Delta u \cdot  u \ dx - 2 h  \int_{\Omega}  e^{2\rho_E^{\delta}/h} \sum_j (\partial_{x_j} \rho_E^{\delta})(  \partial_{x_j} u) \cdot u \ dx \\
& \qquad - h^2 \int_{\Omega} (\Delta e^{\rho_E^{\delta}/h} ) u \cdot e^{\rho_E^{\delta}/h} u \ dx \\
\overset{(i)}{=}&-h^2 \int_{\Omega} e^{2\rho_E^{\delta}/h} \Delta u \cdot  u \ dx + \int_{\Omega} e^{2\rho_E^{\delta}/h} \vert \nabla \rho_E^{\delta} \vert^2 \vert u \vert^2 \ dx \\
=&\int_{\Omega}  e^{2\rho_E^{\delta}/h} (\vert \nabla \rho_E^{\delta} \vert^2  - V +E')u \cdot u  \ dx + \int_{\Omega}  e^{2\rho_E^{\delta}/h} (H_h - E') u \cdot u \ dx. \end{split}
\end{equation}
In (i), we used that 
\[\begin{split} &2 h  \int_{\Omega}  e^{2\rho_E^{\delta}/h} \sum_j (\partial_{x_j} \rho_E^{\delta})(  \partial_{x_j} u) \cdot u \ dx+ h^2 \int_{\Omega} (\Delta e^{\rho_E^{\delta}/h} ) u \cdot e^{\rho_E^{\delta}/h} u \ dx \\
&= h  \int_{\Omega}  e^{2\rho_E^{\delta}/h} \sum_j (\partial_{x_j} \rho_E^{\delta})\partial_{x_j}\vert u \vert^2 \ dx+ \sum_{j}  \int_{\Omega} (\vert \partial_{x_j}\rho_E^{\delta}\vert^2 + h \partial_{x_j}^2 \rho_E^{\delta})e^{2\rho_E^{\delta}/h} \vert u\vert^2 \ dx \\
&= -h  \int_{\Omega}  \sum_j \partial_{x_j}(e^{2\rho_E^{\delta}/h} \partial_{x_j} \rho_E^{\delta})\vert u \vert^2 \ dx+ \sum_{j}  \int_{\Omega} (\vert \partial_{x_j}\rho_E^{\delta}\vert^2 + h \partial_{x_j}^2 \rho_E^{\delta})e^{2\rho_E^{\delta}/h} \vert u\vert^2 \ dx \\
&= -h  \int_{\Omega}   e^{2\rho_E^{\delta}/h} \vert \nabla  \rho_E^{\delta}\vert^2 \vert u \vert^2 \ dx.
\end{split}\]
We conclude from \eqref{eq:equations}
\[ \begin{split}
&h^2 \int_{\Omega} \vert \nabla(e^{\rho_E^{\delta}/h}u)\vert^2 \ dx + \int_{\Omega}  e^{2\rho_E^{\delta}/h} (\lambda_1( V) -\vert \nabla \rho_E^{\delta} \vert^2  -E') \vert u \vert^2  \ dx \\
&\le \int_{\Omega}  e^{2\rho_E^{\delta}/h} (H_h - E') u \cdot u \ dx \\
&\le \Vert \frac{1}{F} e^{\rho_E^{\delta}/h} (H_h-E') u \Vert^2 + \frac{1}{4} \Vert F e^{\rho_E^{\delta}/h} u \Vert^2 \\
&\le \Vert \frac{1}{F} e^{\rho_E^{\delta}/h} (H_h-E') u \Vert^2 + \frac{1}{2} \Vert F_+ e^{\rho_E^{\delta}/h} u \Vert^2+\frac{1}{2} \Vert F_- e^{\rho_E^{\delta}/h} u \Vert^2.
\end{split}\]
Writing $F = F_++F_-$ where 
\[F_+^2 - F_-^2 = \lambda_1( V) -\vert \nabla \rho_E^{\delta} \vert^2  -E \le \lambda_1( V) -\vert \nabla \rho_E^{\delta} \vert^2  -E' ,\]
we find 
\[h^2 \Vert \nabla(e^{\rho_E^{\delta}/h}u)\Vert^2 + \frac{1}{2} \Vert e^{\rho_E^{\delta}/h}  F_+ u  \Vert^2\le \Vert \frac{1}{F} e^{\rho_E^{\delta}/h} (H_h-E') u \Vert^2  +\frac{3}{2} \Vert F_-  e^{\rho_E^{\delta}/h} u \Vert^2.
\]
\end{proof}
\begin{rem}
On unbounded domains the exponential containing the Agmon distance is possibly unbounded.
    Using an approximation argument the proof of this estimate carries over to $\Omega=\mathbb R^n$ where the right-hand side may be infinite in this case, see the discussion in \cite[Remark 1.2]{HS84}.
\end{rem}

Proposition \ref{prop:Agmon} is effective when applied to our single well operator with effective potential $V_{\text{well}}:=V + (1-\chi)$, where $\chi$ is defined after equation \eqref{eq:single-well}.
Let $H_{\text{well}}=-h^2\Delta + V_{\text{well}}: H^2(\RR^2)  \to L^2(\mathbb R^2)$ and $U:=\{x \in \mathbb R^2 ; E \ge \lambda_1(V_{\text{well}}(x))\},$ the well. On the set $M:=\{ x \in \mathbb R^2; d_{d_E}(U,x)>\delta\}$, where $d_{d_E}$ is the distance measure induced by the Agmon distance. 
We assume that $E$ is small enough such that $\mathbb R^2 \setminus M$ is a neighbourhood of the single well located in a small compact neighbourhood of zero. We set $F_+:=\sqrt{\lambda_1(V_{\text{well}}(x))-E-\vert \nabla \rho_E^{\delta} \vert^2}$ which is, by \eqref{eq:estimate}, uniformly bounded away from zero on $M.$ We can extend $F_+$ thereby to a strictly positive function on $\mathbb R^2.$ Since we ask $F_+^2 - F_-^2 = \lambda_1( V_{\text{well}}(x)) -\vert \nabla \rho_E^{\delta} \vert^2  -E,$ we see that $F_-$, has its support in the compact set $\mathbb R^2 \setminus M=\{ x \in \mathbb R^2; d_{d_E}(U,x) \le \delta\}.$ Thus, we conclude that at first that $F_-$ is a bounded function in addition that $F$ is uniformly bounded away from zero.

We have thus obtained
\begin{corr}
\label{corr:well}
Let $u$ be an eigenfunction to $H_{\text{well}}=-h^2\Delta + V_{\text{well}}$ with eigenvalue $E'\le E$, then for some $C_{\delta}>0$
\[h^2 \Vert \nabla(e^{\rho_E^{\delta}/h}u)\Vert^2 +  \Vert e^{\rho_E^{\delta}/h}  u  \Vert^2\le C_{\delta}\Vert  e^{\delta/h} u \Vert^2.
\]
\end{corr}

\section{Exponentially narrow bands}
\label{sec:exp_width}
In this section we study the width of the bands $E_1(h,\mathbf k)\le E_2(h,\mathbf k) \le...$ of the Hamiltonian $H_{\mathbf k}.$ 

We then choose an interval $I(h) $ of length $\mathcal O(h)$ by Theorem \ref{theo:harmonic_approximation} that encloses the lowest $N$ eigenvalues 
\[\begin{split}
&\Spec(H_{\mathbf k})\cap I(h) = \{E_1(h,\mathbf k),..,E_{N}(h,\mathbf k)\} \text{ with } (H_{\mathbf k}-E_j(h,\mathbf k)) u_j(h,\mathbf k) = 0 \text{ and }\\
&\Spec(H_{\text{well}}) \cap I(h) = \{E_1(h),..,E_N(h)\}  \text{ with } (H_{\text{well}}-E_j(h)) \varphi_j(h)= 0,
 \end{split}\]
 where $(u_j(h,\mathbf k))_j$ are an orthonormal system in $L^2(\mathbb R^2/\Gamma)$ and $(\varphi_j(h))$ are orthonormal system in $L^2(\RR^n).$

The goal of this section is to show that low-lying bands of $H_{\mathbf k}$ are of exponentially small width and thus prove Theorem \ref{theo:exp_narrow_bands}.

\begin{proof}[Proof of Theo.\,\ref{theo:exp_narrow_bands}]
We define $t_{\gamma}\varphi(x):=\varphi(x+\gamma)$ where $\gamma \in \Gamma$.
Let $\chi_0 \in C_c^{\infty}(\mathbb R^2 \setminus M;[0,\infty))$ with $\chi_0\vert_{\overline{U}}=1.$ We then define $v_{\mathbf k,j}(h) \in H^2(\mathbb R^2/\Gamma)$ by
\[ v_{\mathbf k,j}(h) := \sum_{\gamma \in \Gamma} e^{-i \mathbf k (\gamma + x)} t_{\gamma}(\chi_{0} \varphi_j(h)).\]
We also compute the overlap 
\begin{equation}
\label{eq:overlap}
 \langle v_{\mathbf k,i}, v_{\mathbf k,j} \rangle_{L^2(\mathbb R^2/\Gamma)} = \int_{\mathbb R^2/\Gamma} \chi_0^2\varphi_i  \overline{\varphi_j} \ dx +\int_{\mathbb R^2/\Gamma} \sum_{\gamma \neq \gamma'} e^{i(\gamma-\gamma')\mathbf k} t_{\gamma}(\chi_{0} \varphi_i) t_{\gamma'}(\chi_{0} \overline{ \varphi_j}) \ dx
 \end{equation}
 showing that for some $C>0$
\begin{equation}
\label{eq:overlap2}
 \langle v_{\mathbf k,i}, v_{\mathbf k,j} \rangle_{L^2(\mathbb R^2/\Gamma)}  = \delta_{i,j} + \mathcal O(e^{-C/h}).
  \end{equation}

We also observe that $H_{\mathbf k} v_{\mathbf k,i}= E_i(h) v_{\mathbf k,i}+r_{\mathbf k,i}$ with $\Vert r_{\mathbf k,i} \Vert = \mathcal O(e^{-C/h})$ for possibly different $C>0$. This follows from
\[\begin{split} H_{\mathbf k} v_{\mathbf k,i}&= E_i(h) v_{\mathbf k,i} +\sum_{\gamma \in \Gamma} e^{-i\mathbf k (\gamma+x)} [-h^2\Delta^2,t_{\gamma}(\chi_{0})] \varphi_i(h,x+\gamma) \\
& \quad -\sum_{\gamma \in \Gamma} e^{-i\mathbf k (\gamma+x)} t_{\gamma}((1-\chi)\chi_0 \varphi_i)
\end{split}\]
where $(1-\chi)\chi_0 \equiv 0$ by choosing the support suitably, and applying Corollary \ref{corr:well} to \[ [-h^2\Delta,t_{\gamma}(\chi_{0})] = -2h^2 \nabla t_{\gamma}(\chi_{0}) \cdot \nabla - h^2 \Delta t_{\gamma}(\chi_{0})\]
which is supported away from the well and this leads to exponentially small corrections.

Self-adjointness and Theorem \ref{theo:harmonic_approximation} imply that (cf.~\cite[Theorem 12.7]{zw12})
\[ 
d(E_i(h,\mathbf k),E_i(h))=\mathcal O(e^{-C/h}),\]
which is the exponential closeness of the corresponding eigenvalues of $H_{\mathbf k}$ and $H_{\text{well}}.$
\end{proof}

\smallsection{Acknowledgements} S.B. acknowledges support from SNF Grant PZ00P2 216019. M.Y. was partially supported by the by the NSF grant DMS-1952939 and by the Simons Targeted Grant Award No. 896630.

\end{document}